\documentclass[sigconf]{acmart}
\usepackage{booktabs} % For formal tables

\usepackage{verbatim}
\usepackage{booktabs}
\usepackage{makecell}
\usepackage{hyphenat} % for preventing hyphenation, using environment \nohyphens{(text block)}

\usepackage[ruled,vlined,linesnumbered]{algorithm2e}

\usepackage{pgf}
\usepackage{algpseudocode}
\newcommand*{\vcenteredhbox}[1]{\begingroup
\setbox0=\hbox{#1}\parbox{\wd0}{\box0}\endgroup}

\DeclareMathOperator*{\argmax}{arg\,max}

\usepackage{times}

\usepackage{soul}
\usepackage{url}
\usepackage{graphicx}
\usepackage{amsmath}
\usepackage{booktabs}
\usepackage[ruled]{algorithm2e}
\usepackage{balance} % for balancing columns on the final page
\usepackage[T1]{fontenc}    % use 8-bit T1 fonts
\usepackage{hyperref}
\usepackage{wrapfig}
\usepackage{subfigure}
\usepackage{fancybox}
\usepackage{tikz}
\usepackage{caption}
\usepackage{algpseudocode}
\usepackage{algorithmicx}
\usepackage[ruled,vlined,linesnumbered]{algorithm2e}
\usepackage{diagbox}
\usepackage{multirow}
\usepackage{amsfonts}       % blackboard math symbols
\usepackage{nicefrac}       % compact symbols for 1/2, etc.
\usepackage{microtype}      % microtypography

\usepackage{graphbox}

\usepackage{graphics}
\usepackage{hyperref}

\copyrightyear{2023}
\acmYear{2023}
\setcopyright{acmlicensed}
\acmConference[SAC '23]{The 38th ACM/SIGAPP Symposium on Applied Computing}{March 27–March 31, 2023}{Tallinn, Estonia} 
\acmBooktitle{The 38th ACM/SIGAPP Symposium on Applied Computing (SAC '23), March 27–March 31, 2023, Tallinn, Estonia}
\acmPrice{15.00}
\acmDOI{10.1145/3555776.3577642} 
\acmISBN{978-1-4503-9517-5/23/03}

\begin{document}
\title{A Hierarchical Game-Theoretic Decision-Making for Cooperative Multi-Agent Systems Under the Presence of Adversarial Agents}

% \titlenote{Produces the permission block, and
  % copyright information}
% \subtitle{Extended Abstract}
% \subtitlenote{The full version of the author's guide is available as
%   \texttt{acmart.pdf} document}
  
\renewcommand{\shorttitle}{SIG Proceedings Paper in LaTeX Format}

% \textit{Qin Yang}
\author{Qin Yang}
\authornote{Contact Dr. Yang with his personal email: RickYang2014@gmail.com}
\orcid{0000-0001-5342-1798}
\affiliation{%
  \institution{Department of Computer Science, University of Georgia}
  \streetaddress{415 Boyd Research and Education Center
University of Georgia}
  \city{Athens} 
  \state{Georgia, U.S.} 
  \postcode{30602-7404}
}
% \email{RickYang2014@gmail.com}
\email{qy03103@uga.edu}

\author{Ramviyas Parasuraman}
% \authornote{The secretary disavows any knowledge of this author's actions.}
\affiliation{%
  \institution{Department of Computer Science, University of Georgia}
  \streetaddress{415 Boyd Research and Education Center
University of Georgia}
  \city{Athens} 
  \state{Georgia, U.S.} 
  \postcode{30602-7404}
}
\email{ramviyas@uga.edu}

% \authornote{This author is the
%   one who did all the really hard work.}
% \affiliation{%

% \affiliation{%
% \affiliation{%
% \affiliation{%
%   \country{China}
% }

% \affiliation{%


% \author{Julius P.~Kumquat}

% \renewcommand{\shortauthors}{B. Trovato et al.}

\begin{abstract}
Underlying relationships among Multi-Agent Systems (MAS) in hazardous scenarios can be represented as Game-theoretic models. 
This paper proposes a new hierarchical network-based model called Game-theoretic Utility Tree (GUT), which decomposes high-level strategies into executable low-level actions for cooperative MAS decisions. It combines with a new payoff measure based on agent needs for real-time strategy games.
We present an Explore game domain, where we measure the performance of MAS achieving tasks from the perspective of balancing the success probability and system costs.
We evaluate the GUT approach against state-of-the-art methods that greedily rely on rewards of the composite actions.
Conclusive results on extensive numerical simulations indicate that GUT can organize more complex relationships among MAS cooperation, helping the group achieve challenging tasks with lower costs and higher winning rates.
Furthermore, we demonstrated the applicability of the GUT using the simulator-hardware testbed - Robotarium.
The performances verified the effectiveness of the GUT in the real robot application and validated that the GUT could effectively organize MAS cooperation strategies, helping the group with fewer advantages achieve higher performance.

% \footnote{This is an abstract footnote}
\end{abstract}




\maketitle

\section{Introduction}

Multi-Agent Systems (MAS) \cite{wooldridge2009introduction} that cooperate with each other show \textit{Distributed Intelligence}, which refers to systems of entities working together to reason, plan, solve problems, think abstractly, comprehend ideas and language, and learn \cite{parker2007distributed}.
Here, an individual agent is aware of other group members and actively shares and integrates its needs, goals, actions, plans, and strategies to achieve a common goal and benefit the entire group \cite{yang2019self,yang2022self}. They can maximize global system utility and guarantee sustainable development for each group member \cite{shen2004degree}.

Systems with a wide variety of agent heterogeneity and communication abilities can be studied, and collaborative and adversarial issues can also be combined in a real-time situation \cite{stone2000multiagent}. Under the presence of adversarial agents, opponents can prevent agents from achieving global and local tasks, even impair individual or system necessary capabilities or normal functions \cite{jun2003path}.
Combining Multi-Agent cooperative decision-making and robotics disciplines, researchers developed the \textit{Adversarial Robotics,} focusing on autonomous agents operating in adversarial environments \cite{yehoshua2015adversarial}. 
From the agent's needs \cite{yang2020hierarchical} and motivations perspective, we can classify an \textbf{Adversary} into two general categories: \textbf{Intentional adversary} (such as an enemy agent, which consciously and actively impairs the agent's needs and capabilities as depicted in Fig.~\ref{fig:overview}) and \textbf{Unintentional adversary} (like obstacles, which passively threaten agent abilities). In this paper, we focus on the Multi-Agent tasks in the presence of adversarial agents (physical and intentional adversaries present at random locations in the environment).

\begin{figure}[t]
\centering
\begin{minipage}[b]{0.57\linewidth}
%\begin{center}
%\setlength{\abovecaptionskip}{6pt}
%\setlength{\belowcaptionskip}{-1pt}
\setlength{\abovecaptionskip}{3.5pt}
\setlength{\belowcaptionskip}{5pt}
\vcenteredhbox{\includegraphics[width=1\columnwidth]{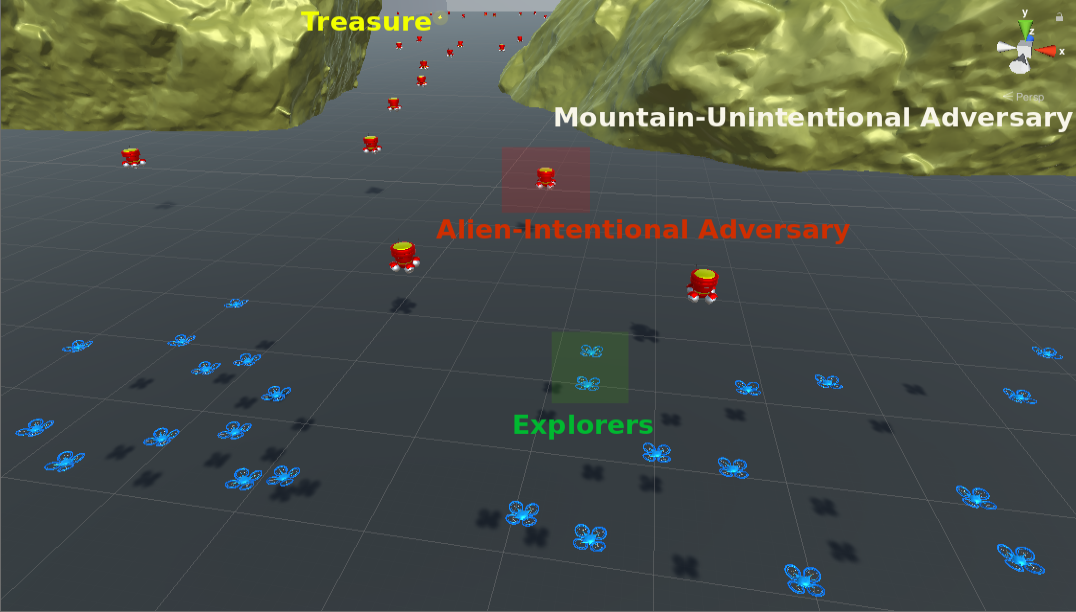}\vspace{-1.5pt}}
\caption*{\small Explore Domain}
\label{fig: explore_game}
%\end{center}
\end{minipage}
\hspace{0mm}
\begin{minipage}[b]{0.18\linewidth}
\setlength{\abovecaptionskip}{3.5pt}
\setlength{\belowcaptionskip}{3.5pt}
\includegraphics[width=1\textwidth]{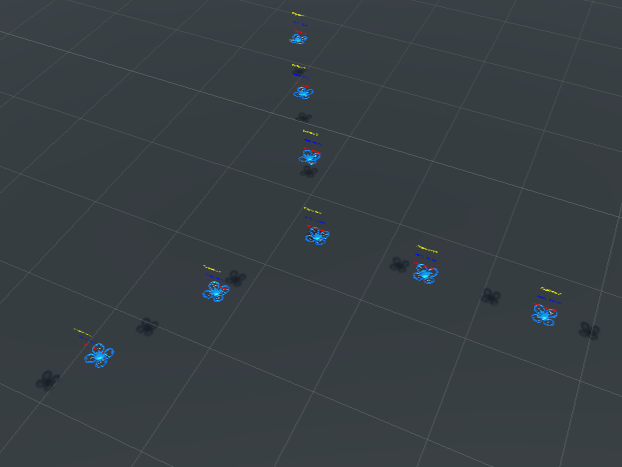}\vspace{0pt}
% \caption{\small{Patrolling}}
\caption*{\small{Patrolling}}
\label{fig: patroling_formation}
\includegraphics[width=1\textwidth]{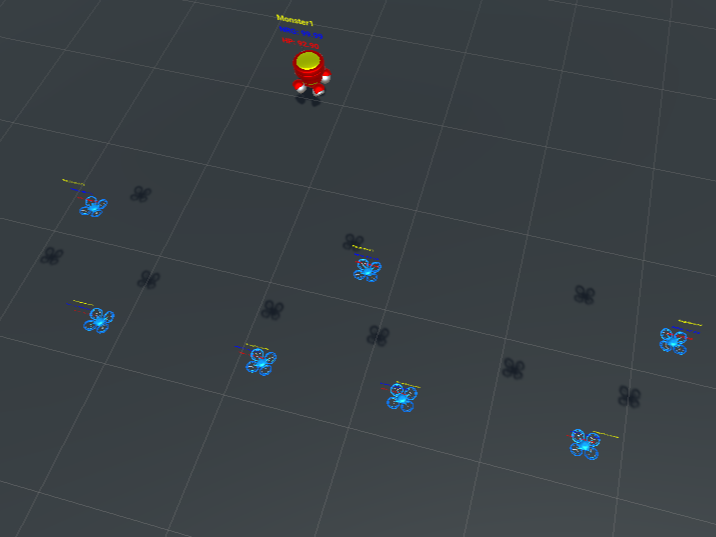}\vspace{0pt}
% \caption{\small{Attacking}}
\caption*{\small{Attacking}}
\label{fig: attacking_formation}
\end{minipage}
\begin{minipage}[b]{0.18\linewidth}
\setlength{\abovecaptionskip}{3.5pt}
\setlength{\belowcaptionskip}{3.5pt}
\includegraphics[width=1\textwidth]{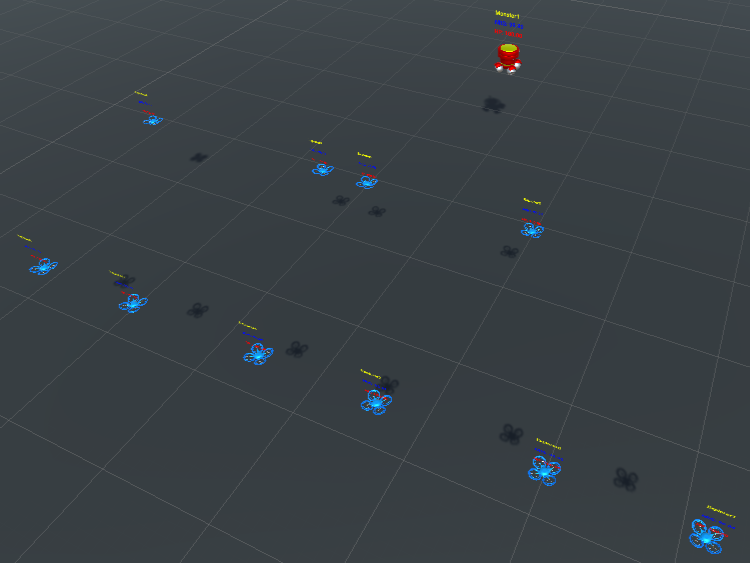}\vspace{0pt}
% \caption{\small{Defending}}
\caption*{\small{Defending}}
\label{fig: defending_formation}
\includegraphics[width=1\textwidth]{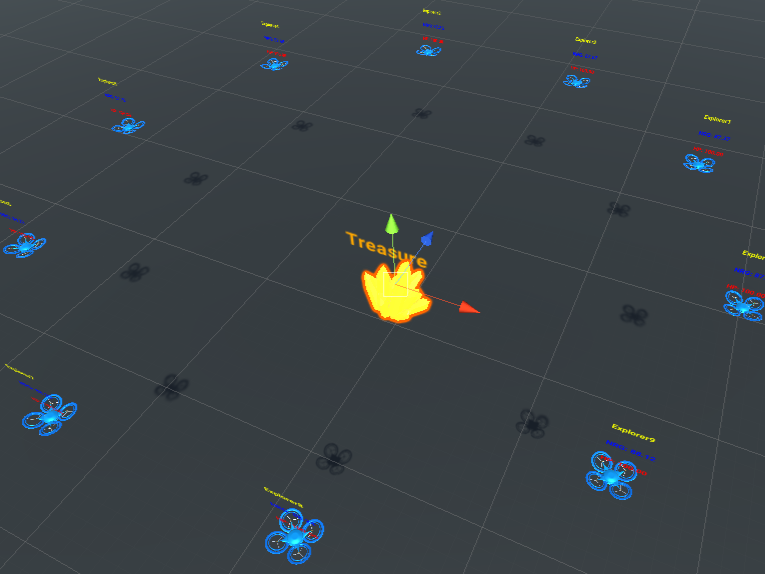}\vspace{0pt}
% \caption{\small{Circling}}
\caption*{\small{Final State}}
\label{fig: coveraging_formation}
\end{minipage}
%\vspace{-6mm}
%\setlength{\abovecaptionskip}{6pt}
% \setlength{\belowcaptionskip}{5pt}
\vspace{-2mm}
\caption{\small{Illustration of the game scenario where the Aliens (opponent agents - peer adversaries) block the paths to a target of the Explorers (protagonist agents). Further, the Explorers cooperate to keep different formations in specific situations depicted on the right side.}}
\vspace{-3mm}
\label{fig:overview}
\end{figure}

MAS research domains focus on solving path planning problems for avoiding static or dynamical obstacles in coverage and patrolling \cite{agmon2011multi,yehoshua2015adversarial} from the unintentional adversary perspective. For intentional (peer) adversaries, the "pursuit domain" \cite{chung2011search} primarily deals with how to guide one or a group of pursuers to catch one or a group of moving evaders \cite{makkapati2019optimal}. Similarly, the robot soccer domain \cite{nadarajah2013survey} deals with how one group of robots wins over another group of robots on a common game element. Foundations for normal-form team games and extensive-form adversarial team games are provided in \cite{von1997team} and \cite{celli2018computational}, respectively. {\color{black} From a game-theoretic control perspective, most of
the research has  focused on single stage games with fixed known agent utilities \cite{browning2005stp}, such as distributed control in communication \cite{zhang2014game} and Multi-Agent task allocation \cite{bakolas2021decentralized}}.
Nevertheless, it is more realistic and practical for MAS to organize more complex relationships and behaviors, achieving given tasks with higher success probability and lower costs in adversarial environments.

\textbf{Contributions:} 
The paper contributes to the field of Multi-Agent\footnote{In this paper, we use the terms agent and robot interchangeably.} systems through the following ways.
\begin{itemize}
\item Firstly, this paper proposes a new hierarchical network model called \textit{Game-theoretic Utility Tree (GUT)} to achieve cooperative decision-making of team-level Multi-Agent strategies under the presence of adversarial agents. 

\item Secondly, we present a game of Explorers vs. Aliens, referred to as \textit{"Explore Domain"} (see Fig.~\ref{fig:overview}), to analyze the MAS performance by balancing the success probability of achieving tasks and system costs by organizing involved agents' relationships and suitable groups' strategies in adversarial environments. 

\item Thirdly, we demonstrate the effectiveness of GUT against the standard decision-making algorithms such as QMIX \cite{pmlr-v80-rashid18a} and greedy approaches in extensive simulations of the \textit{Explore Domain}. The results indicate that \textit{GUT} could organize more complex relationships among MAS, helping the group achieve tasks with low costs and high winning probability. 

\item Finally, to verify the effectiveness of the GUT in the real world (and robotics) applications, we demonstrated the \textit{Explore Domain} in the Robotarium simulator-hardware multi-robot platform \cite{pickem2017robotarium}.
\end{itemize}

\section{Background and Preliminaries}
\label{sec:background}
This section provides a brief background to the Game theory principles used in this paper. Then, we present the agent needs hierarchy based on which the payoff utilities in our game-theoretic approach are designed. Further, we discuss the \textit{"Explore Domain"} problem.
% We then define an adversary agent based on the agent's needs expectations.
%See Appendix~\ref{relative_definition} for other definitions.

\subsection{Game Theory Basics}
% Add a paragraph introducing game-theoretic methods.
% define what is a non-cooperative games (what is a zero sum game)
Game Theory is the science of strategy, which provides a theoretical framework to conceive social situations among competing players and produce optimal decision-making of independent and competing actors in a strategic setting \cite{myerson2013game}.

In a non-cooperative game, players compete individually and try to raise their profits alone. Especially in the zero-sum games, their total value is constant and will not decrease or increase, which means that one player’s profit is associated with another loss. In contrast, if different players form several coalitions trying to take advantage of their coalition, that game will be cooperative \cite{sohrabi2020survey}. \footnote{In this paper, we focus on non-cooperative games, where the ally and enemy agent teams do not cooperate. But, the ally agents will cooperate within the team to decide on a common team strategy by sharing their perception data.}

% \textcolor{red}{What are pure and mixed strategies \\
% What is a normal-form game?}
If a player chooses to take one action with probability 100\%, then the player is playing a pure strategy.
%This is in contrast to a mixed strategy where individual players choose a probability distribution over several actions.
For the game's solutions, if players adopt a \textit{Pure Strategy}, it will provide maximum profit or the best outcome. Therefore, it is regarded as the best strategy for every player of the game. On the other hand, in a \textit{Mixed Strategy}, players execute different strategies with the possible outcome through a probability distribution over several actions. In the game theory, the \textit{Normal Form} describes a game through a matrix, where each player has a set of (mixed) strategies. They select a strategy and play their selections simultaneously. Furthermore, the selection of strategies results in payoff or utility for each player, and its goal in a game is to maximize utility. %TODO%

Furthermore, \textit{Nash Existence Theorem} is a theoretical framework, which guarantees the existence of a set of mixed strategies for a finite, non-cooperative game of two or more players in which no player can improve its payoff by unilaterally changing strategy. It guarantees that every game has at least one Nash equilibrium \cite{jiang2009tutorial}, which means that every finite game has a \textit{Pure Strategy Nash Equilibrium} or a \textit{Mixed Strategy Nash Equilibrium}.
Moreover, in any normal-form game with constant number of strategies per player, an $\epsilon$-\textit{approximate Nash Equilibrium} can be computed in time $\mathop{O}(n^{log~n/\epsilon^2})$, where $n$ is the description size of the game \cite{daskalakis2009complexity}.

\subsection{Agent Needs Hierarchy}
\label{sec:agent_needs}
In \textit{Agent Needs Hierarchy} \cite{yang2020hierarchical}, the abstract needs of an agent for a given task are prioritized and distributed into multiple levels, each of them preconditioned on their lower levels. At each level, it expresses the needs as an expectation over the corresponding factors/features' distribution to the specific group \cite{yang2020needs}. 

Specifically, it defines five different levels of agent needs similar to Maslow's human needs pyramid \cite{yang2021can}. The lowest (first) level is the safety features of the agent (e.g., features such as collision detection, fault detection, etc., that assure safety to the agent, human, and other friendly agents in the environment). The safety needs (Eq.~\eqref{safety_need}) can be calculated through its safety feature's value and corresponding safety feature's probability based on the current state of the agent. After satisfying safety needs, the agent considers its basic needs (Eq.~\eqref{basic_need}), which includes features such as energy levels, data communication levels that help maintain the basic operations of that agent. 
%This second level can be presented as \textit{Basic Needs Expectation}. 
Only after fitting the safety and basic needs, an agent can consider its capability needs (Eq.~\eqref{capability_need}), which are composed of features such as its health level, computing (e.g., storage, performance), physical functionalities (e.g., resources, manipulation), etc. 

At the next higher level, the agent can identify its teaming needs (Eq.~\eqref{teaming_need}) that accounts the contributions of this agent to its team through several factors (e.g., heterogeneity, trust, actions) that the team needs so that they can form a reliable and robust team to successfully perform a given mission. 

Ultimately, at the highest level, the agent learns some skills/features to improve its capabilities and performance in achieving a given task, such as Reinforcement Learning. 
% For instance, an agent may use RL to learn its policy features (e.g., Q table or reward function) using which it can execute appropriate actions based on the current state.
The policy features (Q table or reward function) are accounted into its learning needs expectation (Eq.~\eqref{learning_need}).
The expectation of agent needs at each level are given below:
\vspace{-2mm}
\begin{equation}
\begin{split}
    Safety~~Needs: N_{s_{j}} = \sum_{i=1}^{s_{j}} S_{i} \cdot \mathbb{P}(S_{i}|X_j, T); \label{safety_need}
\end{split}
\end{equation}
\vspace{-2mm}
\begin{equation}
\begin{split}
    Basic~~Needs: N_{b_j} = \sum_{i=1}^{b_{j}} B_{i} \cdot \mathbb{P}(B_{i}|X_j, T, N_{s_{j}}); \label{basic_need}
\end{split}
\end{equation}
\vspace{-2mm}
\begin{equation}
\begin{split}
    Capability~Needs: N_{c_j} = \sum_{i=1}^{c_{j}} C_i \cdot \mathbb{P}(C_i|X_j, T, N_{b_j}); \label{capability_need}
\end{split}
\end{equation}
\vspace{-2mm}
\begin{equation}
\begin{split}
    Teaming~~Needs: N_{t_j} = \sum_{i=1}^{t_j} T_i \cdot \mathbb{P(}T_i | X_j, T, N_{c_j}); \label{teaming_need}
\end{split}
\end{equation}
\vspace{-2mm}
\begin{equation}
\begin{split}
    Learning~~Needs: N_{l_j} = \sum_{i=1}^{l_j} L_i \cdot \mathbb{P(}L_i | X_j, T, N_{t_j}); \label{learning_need}
\end{split}
\end{equation}
{Here,} $X_j=\{P_j,C_j\}$ $\in$ $\Psi$ is the combined state of the agent $j$ with $P_j$ being the perceived information by that agent and $C_j$ representing the communicated data from other agents. $T$ is the assigned task (goal or objective). $S_i$, $B_i$, $C_i$, $T_i$, and $L_i$ are the utility values of corresponding feature/factor $i$ in the corresponding levels - Safety, Basic, Capability, Teaming, and Learning, respectively. $s_j$, $b_j$, $c_j$, $t_j$, and $l_j$ are the sizes of agent $j$'s feature space on the respective levels of needs. 

The collective need of an agent $j$ is expressed as the union of needs at all the levels in the needs hierarchy as in Eq.~\eqref{eqn:need-union} \footnote{Each category needs level is combined with various similar needs (expected values) presenting as a set, consisting of individual hierarchical and compound needs matrix $N_j$.}.
\begin{equation}
    N_j = N_{s_j} \cup N_{b_j} \cup N_{c_j} \cup N_{t_j} \cup N_{l_j} 
    \label{eqn:need-union}
\end{equation}

More specifically, the set of agent needs in a Multi-Agent System can be regarded as a kind of motivation or requirements for cooperation between agents to achieve a specific group-level task.

\subsection{Adversarial Agent Definition}
%We define an adversary agent as follows:
A friendly (ally) agent can contribute to the team, decreasing the individual needs of the team members, while an adversary can harm the team, increasing the overall needs of every team member. Based on this concept, we define an agent $R_1$ as adversary or friendly with respect to an agent $R_2$ as follows. For a certain state $\psi$ $\in$ $\Psi$, the agent $R_1$ is fulfilling a task $T$. Supposing the current needs of $R_1$ is $N_{{R_1}}(\psi, T)$. Considering another agent $R_2$ entering the $R_1$'s observation space, the needs of $R_1$ can be represented as $N_{{R_1}}(\psi \cup R_2, T)$ under the presence of the agent $R_2$. The following equations define the relationship between $R_1$ and $R_2$:
\begin{eqnarray}
    %\underset{i \in Z^+}\argmax(N(\psi, T, a_i)) > \underset{i, j \in Z^+}\argmax(N(\psi, T, a_i, b_j));
    N_{{R_1}}(\psi \cup R_2, T) - N_{{R_1}}(\psi, T)  > 0; & (Adversary)     \label{eqn:adversary}
\\
    N_{{R_1}}(\psi \cup R_2, T) - N_{{R_1}}(\psi, T)  < 0; & (Friendly)
    \label{eqn:friendly} \\
        N_{{R_1}}(\psi \cup R_2, T) - N_{{R_1}}(\psi, T)  = 0. & (Neutral)
    \label{eqn:neutral}
\end{eqnarray}
\newtheorem{myDef}{Definition}
\begin{myDef}[Adversary]
%\begin{myDef}[]
If the needs of $R_1$ increase when $R_2$ is present, then $R_1$ regards $R_2$ as an {Adversary} (Eq. \eqref{eqn:adversary}). \footnote{Note, an obstacle is still an (unintentional) adversary as per this definition, since obstacles will increase the needs of an agent in terms of using more energy to avoid collision risk.}
\label{def:adversary}
\end{myDef}
\begin{myDef}[Friendly]
If the needs of $R_1$ decrease when $R_2$ is present, then $R_1$ sees $R_2$ as a friendly agent (Eq. \eqref{eqn:friendly}). 
\end{myDef}
\begin{myDef}[Neutral]
If the needs of $R_1$ do not change because of $R_2$'s presence, then $R_2$ is neutral to $R_1$ (Eq. \eqref{eqn:neutral}).  
\end{myDef}

\subsection{Explore Domain Problem}

{\color{black}In this paper, to simplify theoretical analysis and numerical calculations, we consider an exemplar problem domain called \textit{Explore domain}, which is described below.}
In \textit{Explore Domain}, $\alpha$ number of agents (called \textbf{Explorers} hereafter) are performing a task $T$, which is to explore an environment and collect rewards by reaching treasure locations. Supposing there are $\beta$ number of (intentional) adversaries (called \textbf{Aliens} hereafter). % and $\gamma$ number of obstacles are randomly distributed in the scenarios.
Explorers can choose a strategy $s_e$ from their strategy space $S_{e}$ and aliens can choose a strategy $s_a$ from their strategy space $S_{a}$. We assume both these strategy spaces are known to the Explorer agents. We also assume that the Aliens do not have a cooperation strategy, and each Alien acts independently on its own.
%Every strategy has corresponding \textit{actions} to execute $s(a_1, a_2, ... , a_r), r \in Z^+$. 

Let $C_i$ represent the system costs of explorer $i$ to perform this task and $W$ denote the success probability (win rate) of the explorer team under the presence of alien(s). We model this as a bi-objective optimization problem (Eq.~\eqref{exploring_game_problem}) of finding an optimal collective team strategy for the explorers $s_e^* \in S_e$ under the premise of maximizing success $W$ (against aliens) while minimizing costs $C$ using the collective needs of the explorers $N_{e}  = \sum_{i=1}^{\alpha} N_i$. 
\begin{equation}
%\begin{split}
   s_e^* = \underset{s_e \in S_e}\argmax [W(s_e | T, S_a, N_e) - \sum_{i=1}^{\alpha} C_i(s_{i} | T, s_e, N_{i})]. %\vspace{-3mm}
   % s.t. & \; \; \; N_{t_i} \geq n_{t_e},~~\forall~{i \in \alpha}.
\label{exploring_game_problem}
%\vspace{-3mm}
%\end{split}
\end{equation}

Without an adversary, the problem shrinks to a typical exploration problem (optimizing $C$ alone) \cite{kim2018multirobot}, and without a task $T$, the problem shrinks to a typical non-cooperative zero-sum game problem (optimizing $W$ alone) \cite{myerson2013game}. 
%\subsection{Application Domains}
%\label{sec:domains}

{\color{black}The proposed problem can be applied to other Real-Time Strategy (RTS) games, such as air combat and StarCarft, and cooperative Multi-Agent/robot mission, like urban search and rescue (USAR) \cite{yang2020needs}, pursuit-evasion game \cite{chung2011search} and robot soccer \cite{nadarajah2013survey}. These domains involve both ally agents and opponent agents (intentional adversaries), and the nature of the strategies the agents or the team can take can allow dissolving the high-level group strategies into primitive or atomic actions.} % \cite{9283249} in the specific mission, such as air combat, StarCarft \cite{vinyals2019grandmaster}, robot soccer, etc.
For instance, in Star Craft \cite{vinyals2019grandmaster}, the strategies a player can make can be composed of the following primitives: What to do? (e.g., Build, Move, Attack, etc.); Who to perform? (which player to execute this action or which opponent to attach, etc.);  Where? (physical location or point of interest); When? (immediately or delayed or how long); etc. 

\section{Proposed GUT Approach}
\label{sec:approach}
In the following description, we ground the description of the decision-making process with GUT using the "Explore domain" as an example.
% , but the approach can be generalized to other relevant domains mentioned earlier.
Fig.~\ref{fig:game_decision_trees} outlines the structure of the \textit{Game-theoretic Utility Tree (GUT)} and its computation units distributed in each level. 
First, the \textit{game-theoretic module} (Fig.~\ref{fig:game_decision_trees} (a)) calculates the nash equilibrium based on the utility values $(u_{11}, ... , u_{nm})$ of corresponding situations, $(p_1, ... , p_{nm})$ presenting the probability of each situation.
Then, through the \textit{conditional probability (CP)} module (Fig.~\ref{fig:game_decision_trees} (b)), the CP of each situation can be described as $(p_{i1}, ... , p_{inm})$, where $p_{inm} = (p_{nm} | p_{i-1}),~i,n,m \in Z^+$. \footnote{{Here,} $p_{i-1}$ and $S_i$ present the probability of previous situation and current strategy in the game-theoretic payoff table; $S_a$, $S_b$ and $n$, $m$ represent their strategy space and size on both sides, respectively.}

\begin{figure}[tbp]
\centering
\includegraphics[width=0.98\columnwidth]{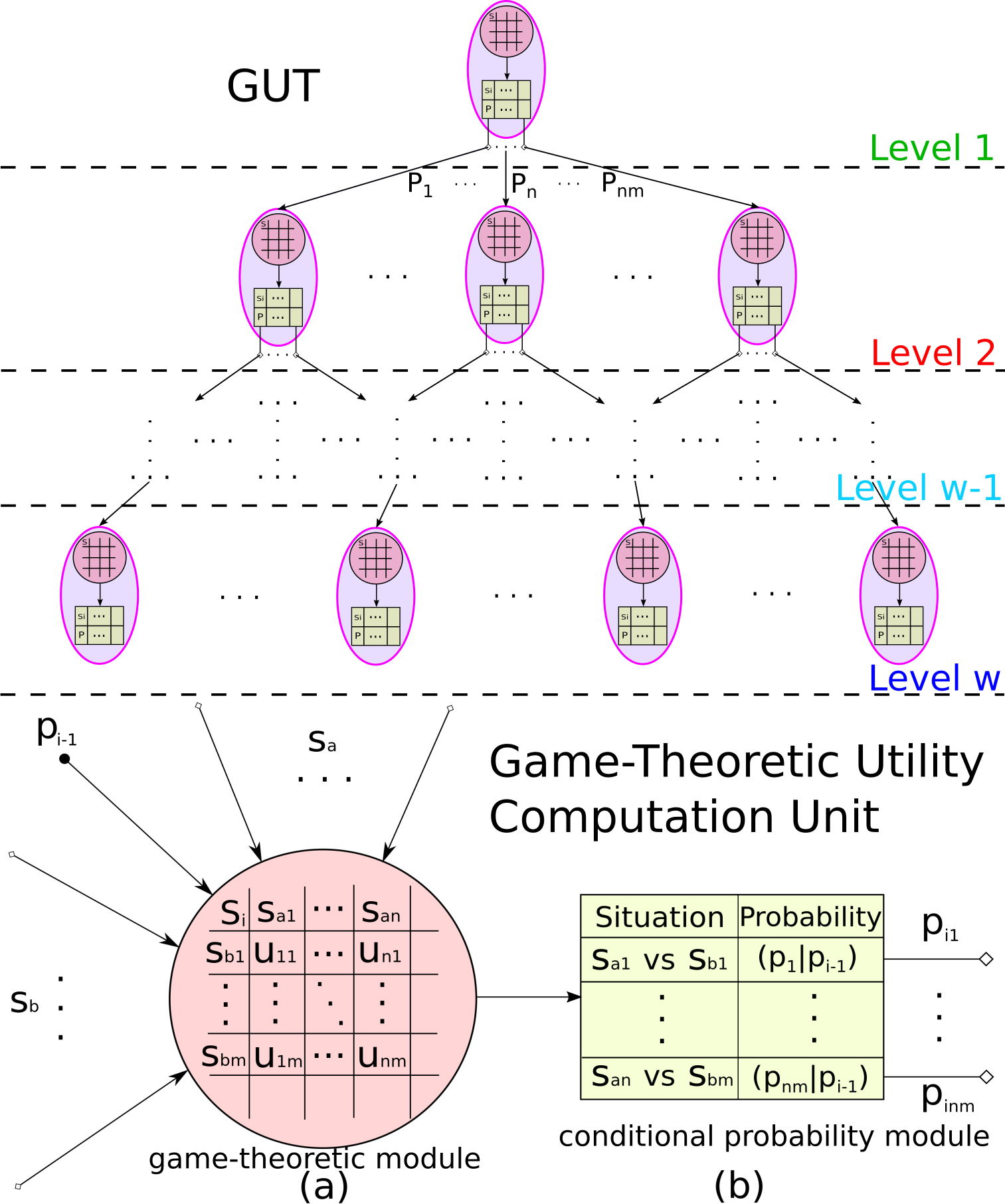}
%\vspace{-2mm}
\caption{Structure of the Game-theoretic Utility Tree (GUT).}
\label{fig:game_decision_trees}
\vspace{-4mm}
\end{figure}

To tackle intentional adversaries, agents first decompose the specific group strategy into several independent sub-strategies based on the same category of individual low-level primitives or atomic operations. Then, through calculating various Nash equilibrium based on different situation utility values in each level's \textit{Game-theoretic Utility Computation Units}, agents can get optimal or suboptimal strategy sets tackling the current status according to \textit{Nash Existence Theorem}. 

From the task execution perspective, the \textit{GUT} can be regarded as a \textit{Task-Oriented Decision Tree}. It decomposes a big game into several sub-games with conditional dependence in each level, which organizes agents' strategies and presents more complex behaviors adapting to adversarial environments. 

We formulate the decision-making problem as a zero-sum game between two groups of agents when both groups are adversaries to each other, as per the Definition~\ref{def:adversary}. 
Let A and B represent the group of ally agents and an adversary agent, respectively. The simultaneous normal-form game representing the non-cooperative game between ally and adversaries is a game structure G=$\langle\{s_A,s_B\},\{N_{A},N_{B}\}\rangle$. 
The theoretical guarantees to prove the effectiveness of the GUT solution for the game G is provided in Theorem~\ref{gut_dec} below.

\newtheorem{proTheo}{Theorem}
\begin{proTheo}[GUT Decision]
\label{gut_dec}
Supposing the GUT for the ally group has w levels of action decomposition together, making a group strategy $s_A = \{s_A^1, s_A^2, ... , s_{A}^w\}$, where $s_A^i$ represents the $i^{th}$ level sub-strategy in the ally group strategy space. Then, we can show that agent group A using GUT against adversary B will have at least one dominant strategy series $s_A^* = (s_A^1*, s_A^2*, ... , s_A^w*)$, which will enable agents A collectively execute an optimal strategy $s_A^*$ against adversary B, promising a solution to the problem in Eq.~\eqref{exploring_game_problem}.
% for the group of explorer agents in GUT.
\end{proTheo}
% \begin{proof}
% see Appendix. \ref{gut_decision}
% \end{proof}
\begin{proof}
\label{proof:gut_decision}
For \textit{w}-level GUT, supposing game G$_i$ is the level \textit{i} of GUT describing the corresponding zero-sum game solved using the payoff (utility) values based on the features in the needs hierarchy corresponding to the level of the GUT. To prove the theorem, we will first prove that at each level $i$, the game-theoretic computation units produce optimal strategies. Then, we will show that by composing strategies at each level, the GUT produces the optimal strategy for the team.

Let the size of (team) strategy space of agents in group A within the level $i$ of GUT is \textit{l$_i$} and that of agent B is \textit{m$_i$}. For GUT-based decision, the \textit{zero-sum} game at each level is 
\begin{equation}
    %G_i=\langle\{A,B\},\{S_e,S_a\},N_{t_{A_i}}\rangle,~~~i \in w;\label{level_game}
    G_i=\langle\{s_A^i,s_B^i\},\{N_{A_i},N_{B_i}\}\rangle, i \in w .  \label{level_game}
\end{equation}
Here, $N_{A_i}$ is the needs/utility values of group A at a level $i$ in the GUT.
%G$_i\langle\{A,B\},\{s_A^i,s_B^i\},\{N_{A}^i,N_{B}^i\}\rangle$, i $\in$ w, 
Based on the needs of agent (Eq.~\eqref{eqn:need-union}), the expectation on the utilities (payoff values) at GUT level $i$ of group A is
\begin{equation}
    \mathbb{E} (U^i) = (u_{gk}^i)_{l_i \times m_i} =  \sum_{j=1}^{|A|} N^j_{A_i} (\psi \cup B, T | s_A^i=g, s_B^i =k).
    \label{need_expectation}
\end{equation}
where, $\psi$ is the combined state perceived by the Multi-Agent group A (with group size $|A|$) executed through the sequence of strategies $s_A$, $1\leq g \leq l_i$ and $1 \leq k \leq m_i$. 
Here, the payoff values depend on the agents' collective needs of the group.

According to \textit{Nash Existence Theorem}, it guarantees the existence of a set of mixed strategies for finite, non-cooperative games of two or more players in which no player can improve his payoff by unilaterally changing strategy. So every finite game has a \textit{Pure Strategy Nash Equilibrium} or a \textit{Mixed Strategy Nash Equilibrium}.
The process can be formalized as two steps:

\textit{a. Compute Pure Strategy Nash Equilibrium}

We can present the agents' utility matrix as Eq. \eqref{A_level1_utility}:\\
\begin{equation}
\left[
\begin{matrix} \label{A_level1_utility}
u_{11} & u_{12} & \cdots & u_{1m_i} \\ 
u_{21} & u_{22} & \cdots & u_{2m_i} \\ 
\vdots & \vdots & \ddots & \vdots \\ 
u_{l_i1} & u_{l_i2} & \cdots & u_{l_im_i} \\ 
\end{matrix} 
\right]
\end{equation}
The row and column correspond to the utilities of agent $A$ and $B$, respectively. We can compute the maximum and minimum values of the two lists separately by calculating each row's minimum value and each column's maximum value.
\begin{equation}
\max \limits_{1\leq k\leq m_i} \min  \limits_{1\leq g\leq l_i} u_{gk} = %TODO Check
\min \limits_{1\leq g\leq l_i} \max  \limits_{1\leq k\leq m_i} u_{gk} \label{p_game_result}
\end{equation}
If the two value satisfy the Eq. \eqref{p_game_result}, we can get the game G$_i$ Pure Strategy Nash Equilibrium (Eq. \eqref{psne}), and corresponding game value in Eq. \eqref{level1_game_result}.
\begin{eqnarray}
&& PSNE=(A_{g^*},B_{k^*}); \label{psne} \\
&& V_{G_i}=u_{g^*k^*} \label{level1_game_result}
\end{eqnarray}

\textit{b. Compute Mixed Strategy Nash Equilibrium}

The probability of choosing the strategy (at level $i$) agent $A$ can be obtained as $X_A=(x_1,x_2,\ldots,x_{l_i})$, satisfying $\sum_{g=1}^{l_i} x_g=1, x_g \geq 0,~g=1,2,\ldots,l_i$.
Similarly, agent $B$ strategies' probability is $Y_B=(y_1,y_2,\ldots,y_{m_i})$, satisfying $\sum_{k=1}^{m_i} y_k=1, y_k \geq 0,~k=1,2,\ldots,m_i$.
Using these probabilities, we can define \textit{(X, Y)} as \textit{Mixed Situation} in certain status. Then, we can deduce the expected utility of agent $A$ and agent $B$ as Eq. \eqref{A_expected_utility} and \eqref{B_expected_utility}, respectively.
\begin{eqnarray}
&& {\mathop{\mathbb{E}}}_A(X,Y)=\sum_{g=1}^{l_i} \sum_{k=1}^{m_i} u_{gk} x_g y_k = \mathop{\mathbb{E}}(X,Y); \label{A_expected_utility}\\
&& {\mathop{\mathbb{E}}}_B(X,Y)=-\mathop{\mathbb{E}}(X,Y) \label{B_expected_utility}
\end{eqnarray}
In the Game G$_i\langle\{S_e,S_a\},N_{{A_i}}\rangle$, if we get all the \textit{Mixed Tactics} of agent $A$ and $B$ as Eq. \eqref{A_tactics} and \eqref{B_tactics}, we can deduce the G$_i$'s \textit{Mixed Expansion} as Eq. \eqref{mixed_expansion}. Furthermore, if a tactic $(X^*, Y^*)$ satisfies Eq. \eqref{A_mne_condition} and \eqref{B_mne_condition}, we define the tactic as the optimal strategy (Eq. \eqref{m_game_result}) in the current state $\psi$.
\begin{eqnarray}
&& S^*_A=\{X_A\}; \label{A_tactics} \\
&& S^*_B=\{Y_B\}; \label{B_tactics} \\
&& G_i^*=\{S^*_A,S^*_B; \mathbb{E}\}; \label{mixed_expansion} \\
&& \mathop{\mathbb{E}}(X^*,Y)\geq V_{s_A}, \forall~Y \in S^*_B; \label{A_mne_condition} \\ %TODO% Check Y in SB
&& \mathop{\mathbb{E}}(X,Y^*)\leq V_{s_B}, \forall~X \in S^*_A; \label{B_mne_condition} \\ %TODO% check X in sA
&& V_{s_A}=V_{G_i}=V_{s_B} \label{m_game_result}
\end{eqnarray}
Therefore, we can prove that at each sublevel in the GUT, we will obtain optimal strategy as per the Nash existence theorem.

GUT decomposes a complex big game into conditionally dependent small games and presents them as a tree structure. We use the Probabilistic Graphical Models \cite{koller2009probabilistic} expressing the GUT computation process. 
Supposing the total number of nodes (sub-games) in the GUT is $\mathcal{N}$, according to the {\it Chain Rule}, the \textit{joint probability} of the GUT in group $A$ can be described as Eq. \eqref{joint_p}.
\begin{equation}
\begin{split}
    \mathbb{P}(X) & = \mathbb{P}(X_1, X_2, ..., X_{\mathcal{N}}) \\
         & = \mathbb{P}(X_1) \mathbb{P}(X_2|X_1) ... \mathbb{P}(X_\mathcal{N}|X_1, X_2, ... , X_{\mathcal{N}-1}) .
        %  & = \mathbb{P}(X_1) P(X_2|X_1) ... \mathbb{P}(X_N|X_{N-(l_1 \times l_2 \times ... \times l_{n-1}) \times (m_1 \times m_2 \times ... \times m_{n-1})}) \\
        %  & = \prod_i \mathbb{P}_i(X_i|Par_G(X_i)), i \in \mathcal{N}
         \label{joint_p}
\end{split}
\end{equation}
Since \textit{Nash Existence Theorem} guarantees that every game has at least one Nash equilibrium \cite{jiang2009tutorial}, we get Eq. \eqref{exist_p}.
\begin{equation}
    % \mathbb{P}_i(X_i) \neq 0 ~~ \Longrightarrow ~ \prod_i \mathbb{P}_i(X_i|Par_G(X_i)) = \mathbb{P}(X) \neq 0
    \mathbb{P}_i(X_i) \neq 0 ~~ \Longrightarrow ~ \mathbb{P}(X) \neq 0, i \in \mathcal{N}
    \label{exist_p}
\end{equation}

This shows that we can always find an optimal strategy against the adversary in the current situation through the GUT by calculating the Nash Equilibrium in the corresponding sub-games distributed at each level.
\end{proof}

% \subsection{Complexity Analysis}
\begin{proTheo}[GUT Complexity]
\label{gut_efficiency}
Assuming the agents' strategies are decomposable into sub-strategies for a specific application domain, then the GUT computes the optimum strategy of an agent (or a group's strategy) more efficiently than applying pure game-theoretic approach applied on the whole strategy space without strategy decomposition (greedy approaches).
\end{proTheo}
\begin{proof}
Using \textit{the master theorem} \cite{cormen2009introduction} of calculating computational efficiency of an Algorithm with a tree-like structure, we will prove the GUT's efficiency. 
Suppose all games at each of the GUT's level has the same size of the strategies space, then the GUT can be described as the running time $T$ of an approach that recursively divides a game $G(\xi)$ of size $\xi$ into $a$ sub-games, each of size $\xi/b,~a,b \in Z^+$ (Eq. \eqref{gut_recur}).
\begin{equation}
    T(\xi) =  a T(\frac{\xi}{b}) + G(\xi),
\label{gut_recur}
\end{equation}
If $G(\xi)$ is the one-level game (without strategy decomposition), the complexity is $\mathop{O}(\xi^{log~\xi/\epsilon^2})$ \cite{daskalakis2009complexity}, for a game with $\epsilon$-approximate near-Nash equilibrium.
% the GUT describes the running time of an approach that recursively divides a game $G(n)$ of size $n$ into $a$ sub-games, each of size $n/b$, where $a$ and $b$ are positive constants. 
Then, the game complexity $G(\xi)$ has the following asymptotic bounds:
\begin{equation}
  \xi^{log_b a} \leq G(\xi) \leq \xi^{log~\xi/\epsilon^2},~\epsilon \in (0,1).
\label{gut_map_2}
\end{equation}
Therefore, it is clear that the GUT with action decomposition is more efficient than a 1-level game of the same size. 
\end{proof}

Furthermore, the scalability of the GUT depends on the efficiency of sharing information across group members, which relies on the specific communication graph between agents.

\begin{table}[h]
     \caption{\small{Level 1 (Attack/Defend) Tactics Payoff Matrix.}}
    \label{tab:first_level_matrix} 
    \begin{minipage}[c]{0.6\columnwidth}
    \begin{center}
    \begin{tabular}{|c|c|c|}
    \hline
    \diagbox{AT}{Utility}{ET} & Attack & Defend \\ \hline
    Attack & $W_{AA}$ & $W_{DA}$ \\ \hline
 Defend & $W_{AD}$ & $W_{DD}$ \\ \hline
 \end{tabular}
\end{center}
\end{minipage}
\begin{minipage}{0.35\columnwidth}
ET - Explorer Tactics

AT - Alien Tactics
\end{minipage}
%\end{table} 
%\begin{table}[t]
%\begin{minipage}{0.4\textwidth}
\vspace{2mm}
     \caption{\small{Level 2 (Who to Attack/Defend) Payoff Matrix.}}
    \label{tab:second_level_matrix}
    \vspace{-2mm}
    \begin{center}
    \begin{tabular}{|c|c|c|c|}
    \hline
    \diagbox{AT}{Utility}{ET} & Nearest & Lowest Ability & Highest Ability \\ \hline
    Nearest & $E_{NN}$ & $E_{A_LN}$ & $E_{A_HN}$ \\ \hline
 Lowest Ability & $E_{NA_L}$ & $E_{A_LA_L}$ & $E_{A_HA_L}$ \\ 
 \hline
 Highest Ability & $E_{NA_H}$ & $E_{A_LA_H}$ & $E_{A_HA_H}$ \\ 
 \hline
 \end{tabular}
\end{center}
%\end{minipage}
%\end{table}
%\begin{table}[t]
%\begin{minipage}{0.4\textwidth}
\caption{\small{Level 3 (How to Attack/Defend) Payoff Matrix.}}
    \label{tab:third_level_matrix}
    \vspace{-2mm}
\begin{center}
   \begin{tabular}{|c|c|c|c|}
    \hline
    \diagbox{AT}{Utility}{ET} & One Group & Two Group & Three Group \\ \hline
    Independent & $HP_{1I}$ & $HP_{2I}$ & $HP_{3I}$ \\ \hline
 Dependent & $HP_{1D}$ & $HP_{2D}$ & $HP_{3D}$ \\ \hline
 \end{tabular}
 %\vspace{3mm}
\end{center}
\vspace{-2mm}
%\end{minipage}
\end{table}

\subsection{GUT Implementation in the Explore Domain}
\label{sec:domain}

In this game, if the explorers do not detect any threat (aliens), they will group in \textit{patrol} formation to explore the circumstance until they reach the treasure's location and \textit{circle} around it (see Fig. \ref{fig:overview}). In the whole process, explorers present a kind of global behaviors using \textit{collective rationality} and caring about their \textit{group interest}. In contrast, aliens are more powerful than explorers but show only \textit{self-interest} and do not cooperate within themselves.
%(See Appendix. \ref{relative_definition} for relative definitions). 

For the gut implementation, we decompose the team strategy into three levels.
At the highest strategy level, the explorer agents decide "what" to do (attack or defend) under the presence of an adversary, using their \textit{teaming needs} as the utility function expressed through a win probability function  $W_{xx}$ for a specific Attach/Defend strategy combination. This helps make group-aware decisions to maximize the chance of collectively reaching the treasure as a team while minimizing the overall team costs. See Table~\ref{tab:first_level_matrix} for the payoff matrix at Level 1.
At the second strategy level (deciding "who" to attack or defend against), the explorers use their collective basic needs expressed as a function of their current energy level $E_{xx}$ in their payoff table. This helps the decision energy-aware (Table~\ref{tab:second_level_matrix}). 
At the lowest strategy level (deciding "how" to attack/defend against), the explorers use their collective safety needs expressed through their health status (\textit{HP} value) to calculate the payoff $HP_{xx}$ at this level (Table~\ref{tab:third_level_matrix}). This ensures the decision is safety-aware since safety is the highest priority of needs as per the needs hierarchy defined in Sec.~\ref{sec:agent_needs}. 
The design of the utility functions at each level is critical to determine whether an agent can calculate reasonable tactics. 

Furthermore, we consider the \textit{Winning Utility} following \textit{Bernoulli Distribution} to represent individual high-level expected utility (teaming \& cooperation needs) in the first level (Eq. \eqref{winning_possibility}). And we assume that the second level's utility is described as the relative \textit{Expected Energy Cost} Eq. \eqref{expected_energy_cost} following \textit{Normal Distribution}. At the lowest level, we use the relative \textit{expected HP cost} to describe the expected utility Eq. \eqref{expected_HP_cost} and assume the attacking times $p(\phi)$ follow \textit{Poisson Distribution}. \footnote{Here, $n$ and $m$ represent the number of Explorers and Aliens respectively; $t_{e}$ and $t_{a}$ represent corresponding average energy levels of both sides; $d$ represents the group average distance between two opponents; attacked {\it HP} cost per time is $h$; agent size is $\phi$; k and g represent the number of attacking Explorers and Aliens, respectively; a, b, and c are corresponding coefficients.}
\begin{equation}
   W(t_{e},t_{a},n,m) = a(\frac{t_{e}}{t_{a}})^{\frac{m}{n}}; \label{winning_possibility}
\end{equation}
\begin{equation}
\begin{split}
    E(n,m,d)= b_0 + b_1 \int_{-\infty}^{+\infty} (n - m) x \frac{1}{\sqrt{2\pi}} e^{-\frac{(x-d)^2}{2}} \mathrm{d}x 
    \label{expected_energy_cost}
\end{split}
\end{equation}
\begin{equation}
\begin{split}
    & H(k,g,\phi_e,\phi_m) = c_0 + c_1 (\sum_{i=1}^{+\infty} k h_{m}p(\phi_m) - \sum_{j=1}^{+\infty} gh_{e}p(\phi_e)); \label{expected_HP_cost}
\end{split}
\end{equation}

Specifically, the first level defines the agent's high-level strategies: \textit{Attack} and \textit{Defend}, which are represented as \textit{Triangle} and \textit{Regular~Polygon} formation shapes of the explorer team. Based on the first level decision, they need to decide the specific opponent attacking or defending in the second level. In the last level, explorers choose how many groups to form based on whether aliens follow other adjacent alien behaviors (dependent) or not (independent) \footnote{We assume that agents have three basic tactics: attacking or defending against the adversary, which is either the \textit{nearest} or has the \textit{lowest} or the \textit{highest} attacking ability.}. A video demonstrating the following experiments is available at \url{https://youtu.be/_zE4eh03Voo}.

\section{Numerical Experiments}
\label{sec:evaluation}
% \vspace{-2mm}
Considering cross-platform, scalability, and efficiency of the simulations, we chose the ``Unity'' game engine to simulate the \textit{Explorers and Aliens Game} and selected Gambit \cite{mckelvey2006gambit} toolkit for calculating each level's Nash Equilibrium in the GUT.

In the experiments, we suppose each explorer has the same initial energy and HP levels, and every moving step will cost $0.015\%$ energy. 
Every communication round and per time attacking will cost $0.006\%$ and $0.01\%$ energy level, respectively.
Attacked by the aliens will cost the explorer $0.15\%$ HP per time. 
Per design, the aliens are 3x more powerful and capable than the explorers in the attacks, with per time attacking energy and per time attacked HP costs for an alien are $0.03\%$ and $0.05\%$, respectively.
We simulated the attack/defend process through the agent proximity and the time spent in the zone attacking/defending. 
To simplify our experiments and focus on \textit{Explorers}' high-level interactions and decision-making, we assume that the agent's speed is constant in the simulations and do not consider optimizing the adversary agent's strategy and performance.

\begin{figure*}[t]
\centering
% \setlength{\abovecaptionskip}{0pt}
% \setlength{\belowcaptionskip}{-10pt}
% \hspace{0mm}
\begin{minipage}[b]{0.245\linewidth}
\includegraphics[width=1\textwidth]{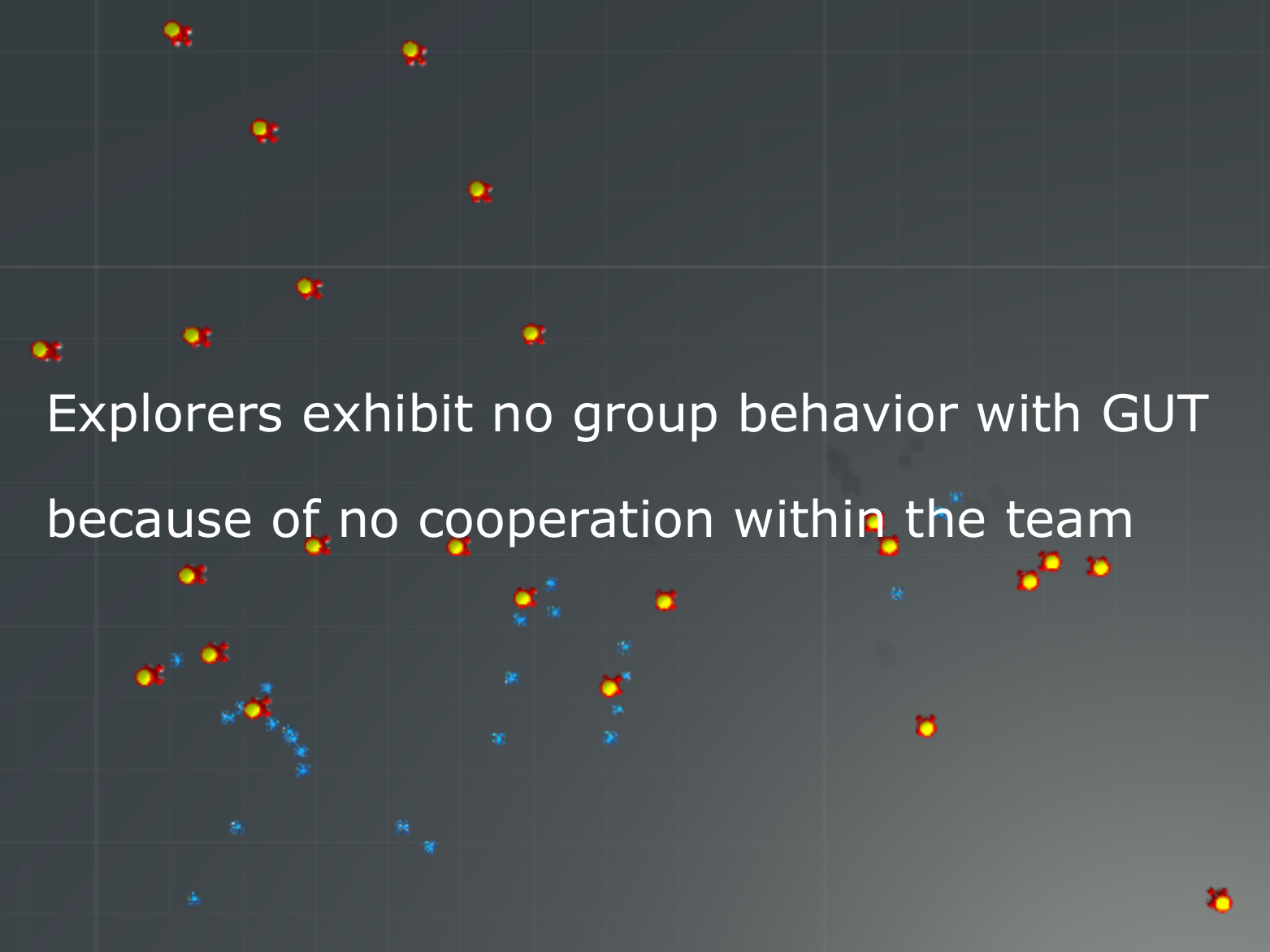}\vspace{0pt}
\caption{\small{GUT (NC)}}
\label{fig: qmix_gut}
\end{minipage}
\begin{minipage}[b]{0.245\linewidth}
\includegraphics[width=1\textwidth]{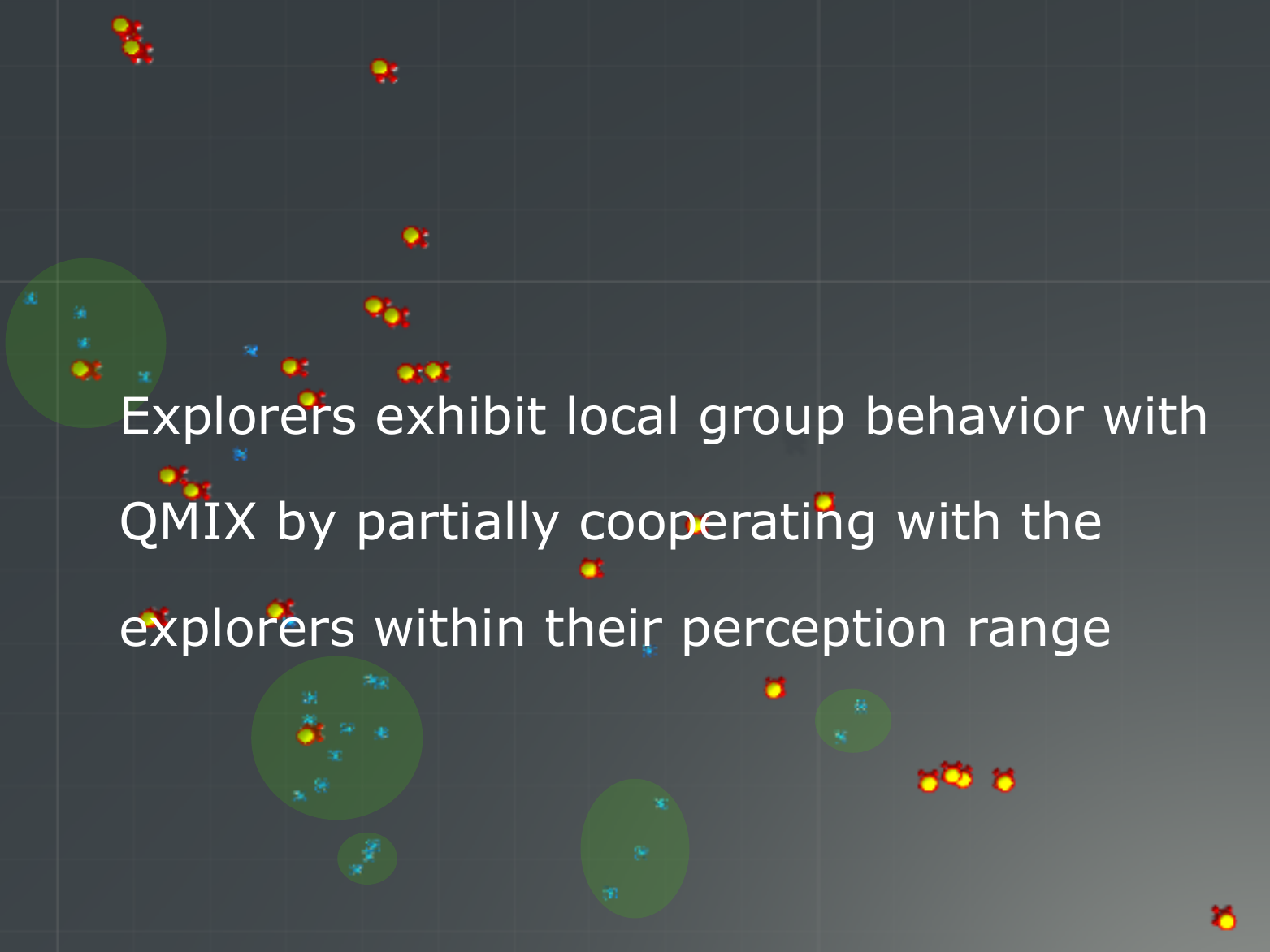}\vspace{0pt}
\caption{\small{Greedy/QMIX (PC)}}
\label{fig: qmix}
\end{minipage}
\begin{minipage}[b]{0.245\linewidth}
\includegraphics[width=1\textwidth]{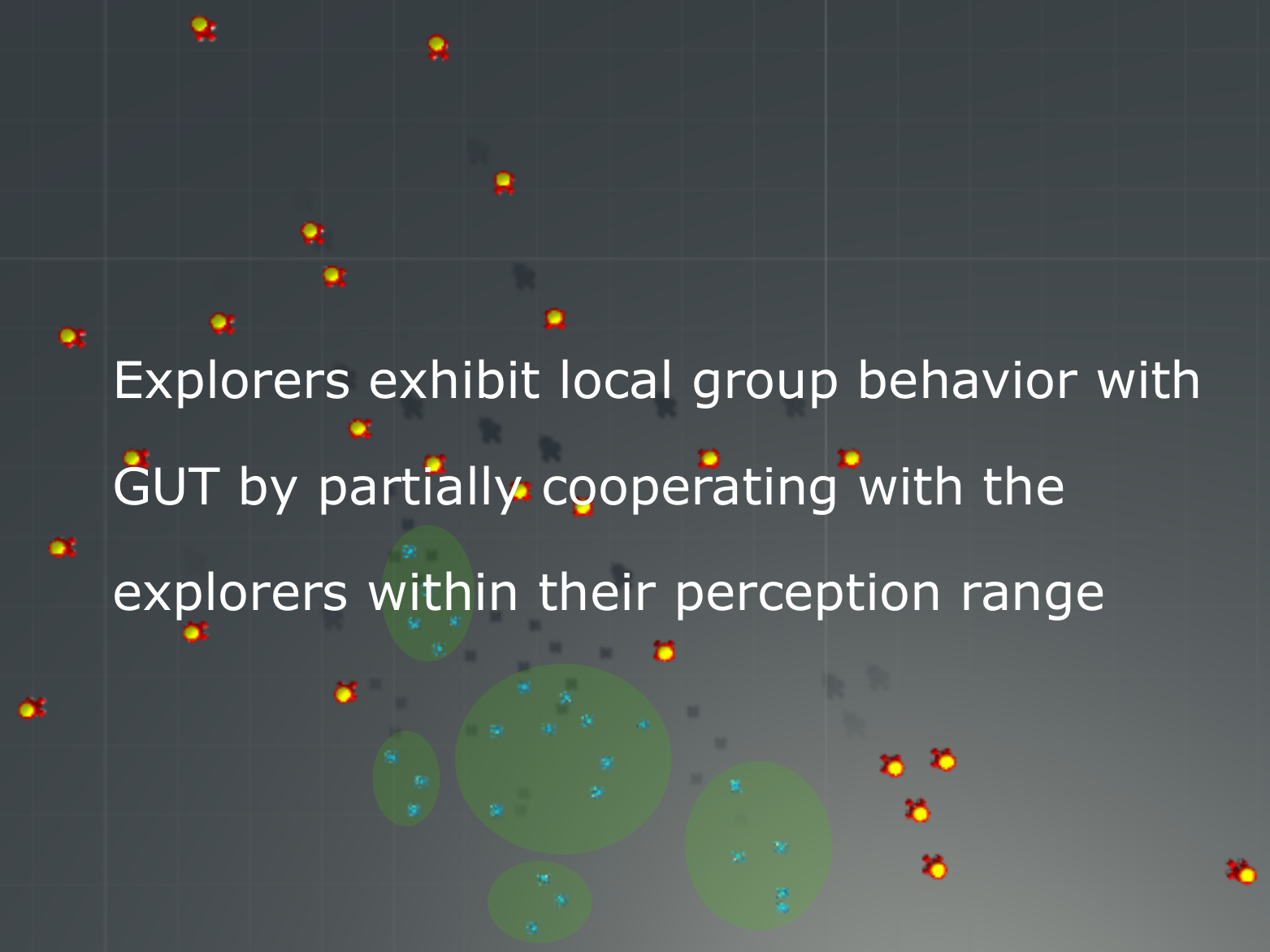}\vspace{0pt}
\caption{\small{GUT (PC)}}
\label{fig: gut_pc}
\end{minipage}
\begin{minipage}[b]{0.245\linewidth}
\includegraphics[width=1\textwidth]{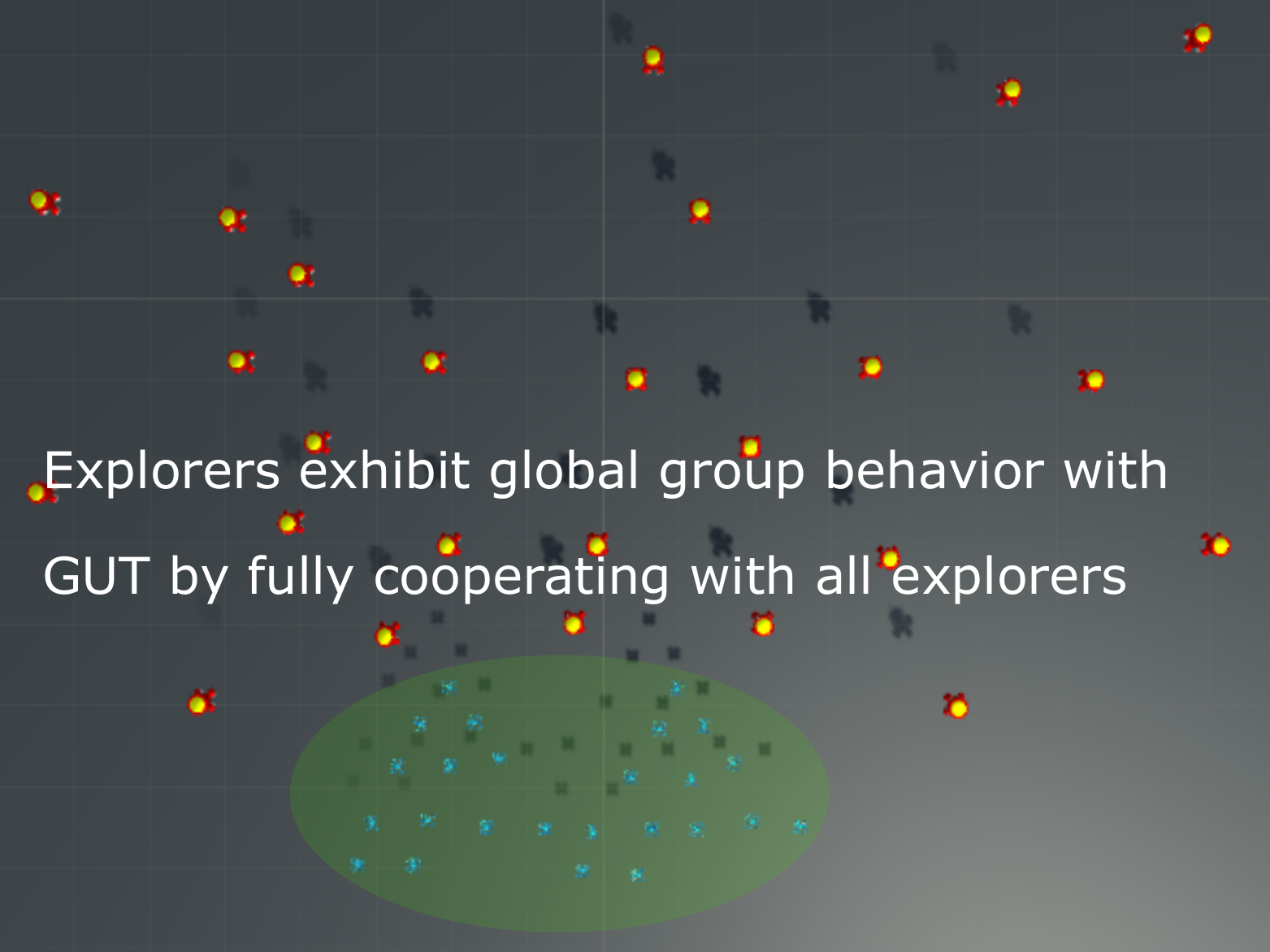}\vspace{0pt}
\caption{\small{GUT (FC)}}
\label{fig: gut_fc}
\end{minipage}
\vspace{-5mm}
\end{figure*}

\subsection{Compared Scenarios and Methods}
\label{sec:methods}

\begin{figure*}[t]
%\vspace{-4mm}
\centering
    \subfigure[Explorer Average HP Cost]{
    \begin{minipage}[t]{0.32\linewidth}
 \includegraphics[width=1\textwidth]{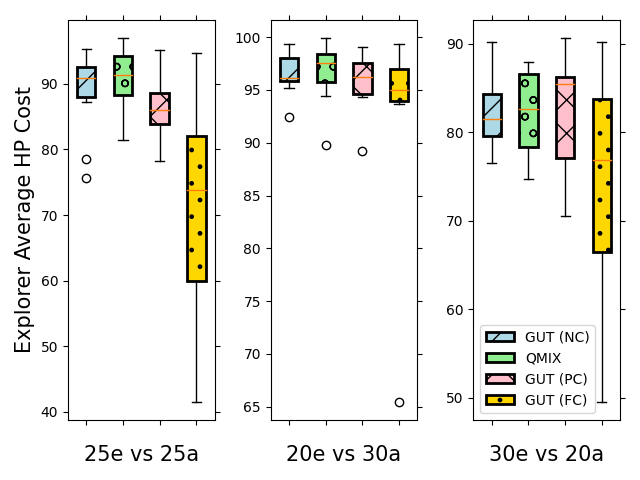}
 \vspace{-2mm}
 \label{fig: eahc}
 \end{minipage}}
 \subfigure[Kill Per Alien - Explorer Lost]{
    \begin{minipage}[t]{0.32\linewidth}
 \includegraphics[width=1\textwidth]{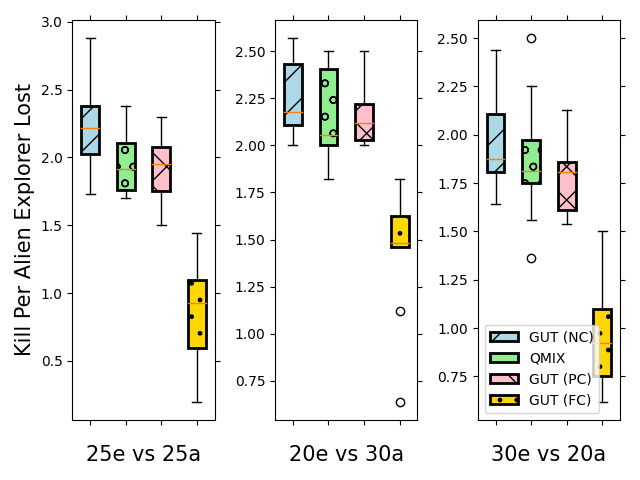}
 \vspace{-2mm}
 \label{fig: kpmel}
 \end{minipage}}
 \subfigure[Kill Per Alien - HP Cost]{
    \begin{minipage}[t]{0.32\linewidth}
    \includegraphics[width=1\textwidth]{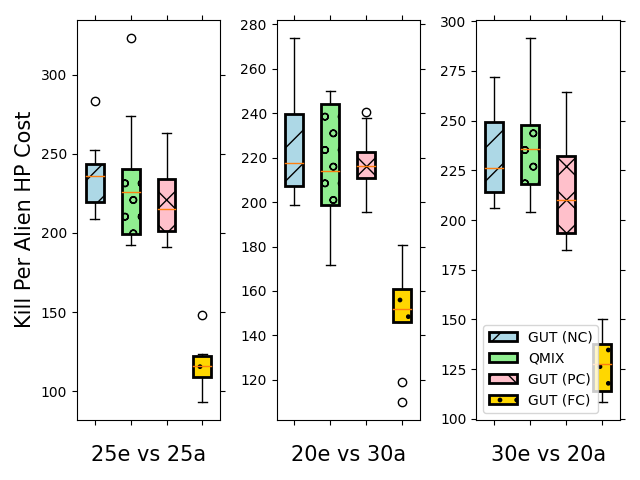}
    \vspace{-2mm}
    \label{fig: kpmhc}
    \end{minipage}}
    \vspace{-3mm}
\caption{\small{The performance of explorers in interaction experiments with different proportions (e-explorers, a-aliens).}}
\label{fig: interactive_experiment}
\vspace{-2mm}
\end{figure*}

We evaluate \textit{GUT} from two different scenarios: 1) \textit{Interaction Experiments} compares the performance of explorers' cooperative strategies between \textit{GUT} and \textit{Greedy/QMIX} methods; 2) \textit{Information Prediction} demonstrates the \textit{GUT} when different predictive models are implemented to estimate aliens' states.
Further, we analyze \textit{GUT} by simulating various cooperation methods, comparing the performance with the state-of-the-art greedy-based approaches - QMIX \cite{pmlr-v80-rashid18a}.

% \vspace{-2mm}
\paragraph*{1) GUT (NC)} [\textit{Noncooperation} - Self-interest]
In this situation, explorers adopt \textit{GUT} computing the winning rate based on its perception, but no communication, which means that it does not get the consistency to attack or defend the specific alien (Fig.~\ref{fig: qmix_gut}).
%\vspace{-2mm}

% \vspace{-2mm}
\paragraph*{2) Greedy/QMIX (PC)} [\textit{Partial Cooperation}]
QMIX is a value-based RL method applied to MAS \footnote{Here, we only focus on the decision-making part of QMIX, which considers the global benefit yielding the same result as a set of individual rewards.}. It allows each agent to participate in a decentralized execution solely by choosing greedy actions for its rewards. Accordingly, we assume that each explorer can cooperate, communicate, and share information with its observing explorers. Then through calculating the corresponding win rate based on the number of observed explorers and aliens, it chooses to attack or defend the specific \textit{hp lowest} target (Fig.~\ref{fig: qmix}).
%\vspace{-2mm}

% \vspace{-2mm}
\paragraph*{3) GUT (PC)} [\textit{Partial Cooperation}]
Here, we consider the same situation as \textit{QMIX}, but explorers calculate the winning rate with \textit{GUT} with consistency in strategic decisions (Fig.~\ref{fig: gut_pc}).
%\vspace{-2mm}

% \vspace{-2mm}
\paragraph*{4) GUT (FC)} [\textit{Full Cooperation} - Collective Rationality]
In this approach, we assume each explorer work and make decisions with \textit{GUT} in full communication mode. This enables every group member to achieve consensus on the  decisions in a distributed system, prioritizing \textit{group's interest} (Fig.~\ref{fig: gut_fc}).

Especially, in partial cooperation, agents cooperate with their neighbors, which they can sense, to work as a team. As shown in Fig. \ref{fig: gut_pc}, the entire group has been divided into five subgroups, and each subgroup only optimizes its expected return within it. In contrast, all agents work as one group in full cooperation, shown in Fig. \ref{fig: gut_fc}, and optimize the whole team' expected return.
\begin{table}[h]
% \vspace{-2mm}
 \caption{\small{Winning Rate Results of the Interaction Experiments.}}
    \label{tab: wrc}
\resizebox{\columnwidth}{!}{\begin{tabular}{|c|c|c|c|c|}
%\scriptsize
    \hline
    \diagbox{PRD}{Winning Rate}{APP} & GUT (NC) & QMIX (PC) & GUT (PC) & GUT (FC)\\ \hline
    20e vs 30a & 40\% & 50\% & 50\% & 70\% \\ 
    \hline
 25e vs 25a & 90\% & 100\% & 100\% & 100\% \\ 
 \hline
 30e vs 20a & 100\% & 100\% & 100\% & 100\% \\ 
 \hline
 \end{tabular}}
% \vspace{-3mm}
\end{table}

\subsection{Interaction Experiments}
In these experiments, the environment is free of obstacles.% (\textit{Unintentional Adversary}) 
%(Fig.~\ref{fig: no_obstacles}). 
We consider three different ratios (A/E) between the number of aliens and explorers as follows: \textit{20 explorers vs 30 aliens}, \textit{25 explorers vs 25 aliens} and \textit{30 explorers vs 20 aliens}. For each scenario, we conduct ten simulation trials for each proportion with the same environment setting \footnote{We assume that an agent can detect opponents' current state in its perception range.}. In these experiments, the aliens follow a self-interest (noncooperation) random action strategy, while the explorers follow either the GUT or Greedy strategy with different cooperation mentioned in Sec.~\ref{sec:methods}.

\paragraph*{Results}

Fig.~\ref{fig: interactive_experiment} presents the performance results in the interaction experiments. We can see that \textit{GUT} (FC) has the best performance compared with other cases in terms of low HP costs and loss of explorers to win against aliens. 
Table. \ref{tab: wrc} shows the winning rate, which also reflects similar results.

Specifically, the GUT (NC), QMIX, and GUT (PC) do not have much difference between in \textit{explorer average HP cost} results (Fig. \ref{fig: eahc}), but in \textit{num. of explorers lost for killing an alien} (Fig.~\ref{fig: kpmel}) and \textit{HP cost for killing an alien} (Fig.~\ref{fig: kpmhc}), the QMIX and GUT (PC) show some advantage comparing with GUT (NC). 
% The rationale for these observations is explained below.
The results show that cooperation conduces to decrease the costs and boost the winning rate for more challenging tasks. More importantly, GUT can help agents represent more complex group behaviors and strategies, such as forming various shapes and separating different groups adapting adversarial environments in MAS cooperation. It vastly improves system performance, adaptability, and robustness. Besides, communication plays an essential role in cooperation, such as solving conflicts and getting consistency through negotiation. In GUT (NC) and QMIX, agents only share local information about the number of observing agents for naive attacking or defending behaviors. However, GUT (FC) presents more complex relationships between agents' cooperation by organizing global communication data.

\begin{figure*}[t]
\centering
\begin{minipage}[b]{0.245\linewidth}
\includegraphics[width=1\textwidth]{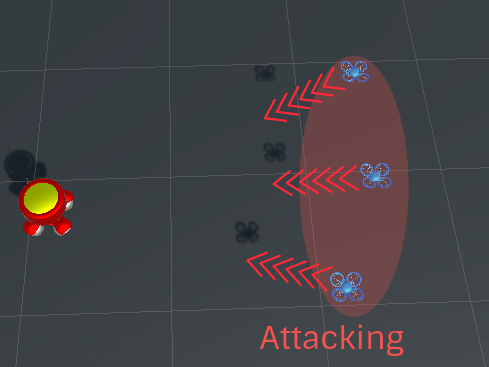}\vspace{0pt}
\caption{\small{One Alien QMIX}}
\label{fig: qmix_compare}
\end{minipage}
\begin{minipage}[b]{0.245\linewidth}
\includegraphics[width=1\textwidth]{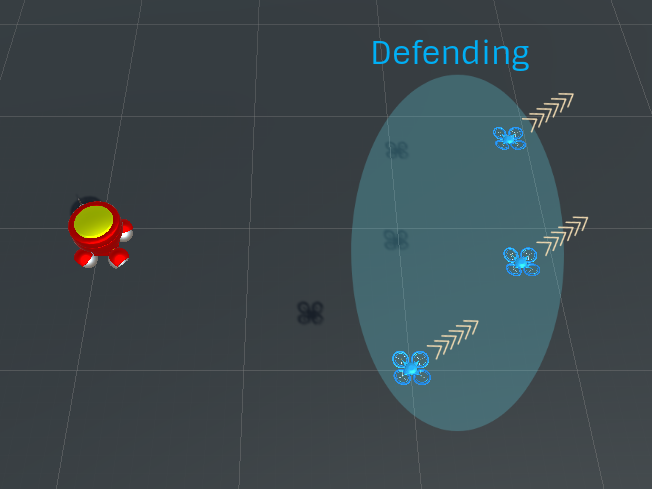}\vspace{0pt}
\caption{\small{One Alien GUT}}
\label{fig: gut_compare}
\end{minipage}
\begin{minipage}[b]{0.245\linewidth}
\includegraphics[width=1\textwidth]{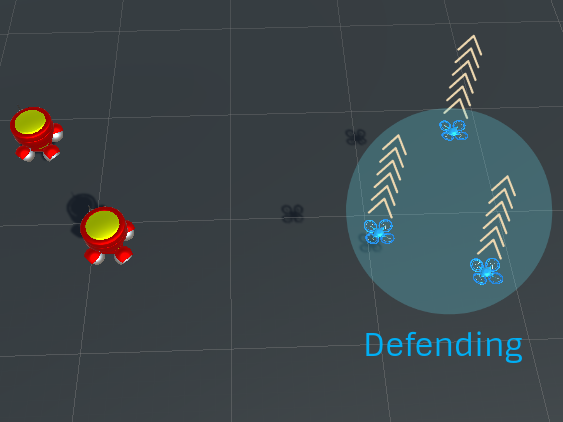}\vspace{0pt}
\caption{\small{Two Aliens QMIX}}
\label{fig: 2qmix_compare}
\end{minipage}
\begin{minipage}[b]{0.245\linewidth}
\includegraphics[width=1\textwidth]{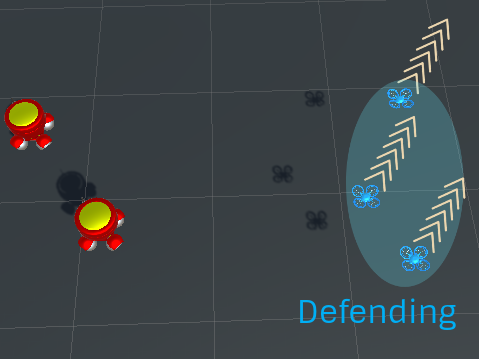}\vspace{0pt}
\caption{\small{Two Aliens GUT}}
\label{fig: 2gut_compare}
\end{minipage}
\vspace{-5mm}
\end{figure*}
\begin{figure*}[t]
\centering
 \subfigure[\small{Explorer Average Energy Cost}]{
    \begin{minipage}[t]{0.32\linewidth}
 \includegraphics[width=1\textwidth]{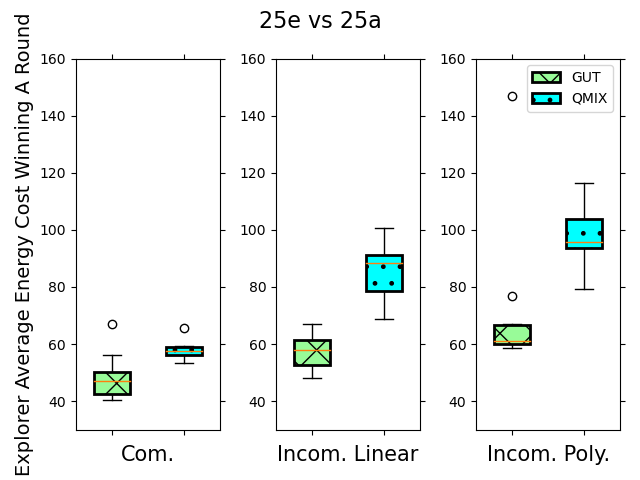}
 \vspace{-2mm}
 \label{fig: kpmel_pre_mqmix}
 \end{minipage}}
     \subfigure[\small{Explorer Average HP (Health) Cost}]{
    \begin{minipage}[t]{0.32\linewidth}
 \includegraphics[width=1\textwidth]{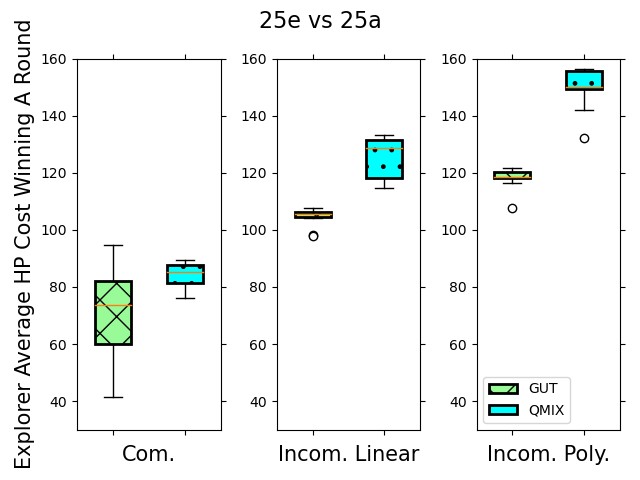}
 \vspace{-2mm}
 \label{fig: eahc_pre_mqmix}
 \end{minipage}}
 \subfigure[\small{Number of explorers lost per winning round}]{
    \begin{minipage}[t]{0.32\linewidth}
    \includegraphics[width=1\textwidth]{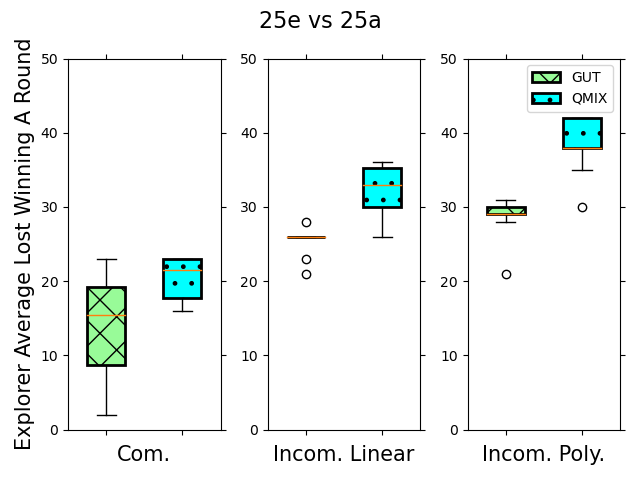}
    \vspace{-2mm}
    \label{fig: kpmhc_pre_mqmix}
    \end{minipage}}
    \vspace{-2mm}
\caption{\small{Explorers' performance results with different predictive models with obstacles in the environment.}}
\label{fig: parameter_estimation_experiment_mqmix}
\vspace{-2mm}
\end{figure*}

To highlight the difference in how the GUT and Greedy/QMIX approaches work, let us consider an example scenario where an explorer team has only two strategies, attack and defend. For the QMIX approach, selecting attacking or defending relates to the number of the partial group and opponents and their attacking ability directly. However, for the GUT, they need to compute the \textit{Nash Equilibrium} on all possible strategies both sides can take. The payoff values are calculated based on relative shared information before selecting a suitable strategy. It means that agents might choose a different strategy in the same condition, comparing QMIX and GUT. 

For example, in a scenario with one adversary, the greedy method chooses to attack the adversary (Fig. \ref{fig: qmix_compare}), while the GUT chooses to defend against the adversary (Fig. \ref{fig: gut_compare}). 
However, in a scenario with two adversaries, QMIX (Fig. \ref{fig: 2qmix_compare}) and GUT (Fig. \ref{fig: 2gut_compare}) have the same strategy reacting to the adversaries. It leads to more costs for the whole system, especially the loss of health (HP), which increases the system's instability.
The two examples explain that instead of choosing greedy actions for group rewards in some scenarios, we need to consider various needs within the team and individual and balance them in a long round. It can improve the sustainability and robustness of the group, especially in uncertain environments.

\subsection{Information Prediction}
\label{sec:prediction}
We design two kinds of perceiving models to analyze individual and system performance. 1) In {\textit{complete information}} scenario, if an agent can perceive an adversary, it can detect the adversary's status, such as the unit attacking energy cost and energy level. 2) In {\textit{incomplete information}} scenario, an agent can not gain opponents' state in its observable range. 
% So, it can obtain the adversary's state through some predictions. 

Assuming that adversary's unit HP cost \textit{HP$_{uc}$} and average system cost \textit{HP$_{asc}$} can be perceived. Then, we use two prediction models (\textit{Linear} regressor - Eq. \eqref{linear_regression} and \textit{Polynomial regressor} - Eq. \eqref{polynomial_regression}) to predict adversaries' unit attacking energy cost $E_{uc}$ and current energy level $E_{el}$ based on the corresponding HP costs of the adversary. \footnote{Here, $\beta$ represents corresponding regression coefficients ($\beta_{uc_{0,1,2}}$ = \{0.08, 0.03, 0.0001\}, $\beta_{asc_{0,1,2}}$ = \{0.03, 0.0003, 0.00001\}), $\varepsilon$ represents a Gaussian noise $\mathcal{N}(0,1)$.}
\begin{equation}
\begin{split}
    & E_{uc} = HP_{uc} \times \beta_{uc_0} + \varepsilon; \\
    & E_{el} = 100 - HP_{asc} \times \beta_{asc_0} + \varepsilon.
    \label{linear_regression}
\end{split}
\end{equation}
\begin{equation}
\begin{split}
    & E_{uc} = HP_{uc}^2 \times \beta_{uc_2} + HP_{uc} \times \beta_{uc_1} + \varepsilon; \\
    & E_{el} = 100 - HP_{asc}^2 \times \beta_{asc_2} -  HP_{asc} \times \beta_{asc_1} + \varepsilon.
    \label{polynomial_regression}
\end{split}
\end{equation}

\subsubsection{Environment without obstacles}
In this scenario, we consider five proportions of explorers and aliens distributed on the map randomly. For each ratio, we also conduct ten simulation trials with the same experimental setting. 
\paragraph*{Results} 
Table~\ref{tab:Strategy_Comparision} shows the results of these experiments. From an agent's perspective, the \textit{Linear Regression} model has more accuracy than the \textit{Polynomial Regression} model, comparing with the result trend of \textit{Complete Information} (ground truth). From a system perspective, the win rate and the mean energy/HP cost with different predictive models show similar results. 
%We provide more data/graphs in Appendix~\ref{sec:noobstacles}.

%\vspace{-4mm}
\subsubsection{Environment with obstacles}
In this scenario, we consider a more complex environment, where there are obstacles (two mountains in the experiment setting) as well as adversaries. 
Here, the aliens adopt the {Greedy/QMIX} approach to make their individual decision, while explorers use either QMIX or GUT. 
We fix the number of explorers (E=25) and aliens (A=25) and conduct ten trials for each predictive model. 
To implement this scenario, we designed an obstacle avoidance algorithm, which can help agents collectively avoid obstacles by adapting their edge's trajectory until it finds a suitable route to the goal. 

\begin{table}[t]
%\vspace{-4mm}
\caption{\small{System Utility Comparison. Ra: Ratio of Explorers to Aliens,  WR: Winning Rate, C$_{s_e/w}$: system average energy cost winning a round, C$_{s_{hp}/w}$}: system average HP cost winning a round.}
\label{tab:Strategy_Comparision}
\begin{center}
\resizebox{\columnwidth}{!}{\begin{tabular}{|c|c|c|c|c|c|c|c|c|c|}
\hline 
\multirow{2}*{Ra} &  \multicolumn{3}{c|}{Complete Info} &  \multicolumn{3}{c|}{Incomplete - Linear Prediction} & \multicolumn{3}{c|}{Incomplete - Poly Prediction}\\
\cline{2-10}
~&WR&C$_{s_e/w}$&C$_{s_{hp}/w}$&WR&C$_{s_e/w}$&C$_{s_{hp}/w}$&WR&C$_{s_e/w}$&C$_{s_{hp}/w}$\\
\hline
\multicolumn{10}{c}{Environment without obstacles} \\
\hline
20:30&70\%&1077.37&2649.80&30\%&2367.14&6306.90&30\%&2726.64&6216.44\\
\hline
20:25&90\%&818.63&2027.98&50\%&1414.84&3807.23&40\%&1375.83&4824.64\\
\hline
25:25&100\%&1211.09&1772.00&90\%&1432.06&2606.12&80\%&1789.15&2949.89\\
\hline
25:20&100\%&1414.35&1739.78&100\%&1449.07&1960.45&100\%&1472.46&2177.52\\
\hline
30:20&100\%&1608.18&2241.09&100\%&2041.85&2370.76&100\%&1961.86&2271.48\\
\hline
\multicolumn{10}{c}{Environment with obstacles} \\
\hline
25:25&100\%&1443.85&2110.91&70\%&2144.10&3143.63&60\%&2451.37&3742.98\\
\hline
\end{tabular}}
\vspace{-4mm}
\end{center}
\end{table}

\paragraph*{Results}
Through the performance results (normalized per explorer) shown in Fig.~\ref{fig: parameter_estimation_experiment_mqmix}, we notice that due to obstacles (unintentional adversaries) involved, agents' average HP and energy cost for winning a round increase distinctly. 
Especially, Fig. \ref{fig: kpmhc_pre_mqmix} describes the average number of explorers sacrificed in the game to win in a round as a team against adversaries. 
The data show that the \textit{Linear Regression} model presents higher accuracy than the {\it Polynomial} model because the current energy and HP calculation models do not involve nonlinear factors. %(see Appendix~\ref{sec:energy-utility} and \ref{sec:hp-utility}).
Table. \ref{tab:Strategy_Comparision} showing the system-level performance results reveal a similar conclusion that obstacles lead to the decrease of the winning rate and more system costs compared to the scenario without obstacles.

\begin{figure*}[t]
\centering
% \setlength{\abovecaptionskip}{0pt}
% \setlength{\belowcaptionskip}{-10pt}
% \hspace{0mm}
\begin{minipage}[b]{0.314\linewidth}
\includegraphics[width=1\textwidth]{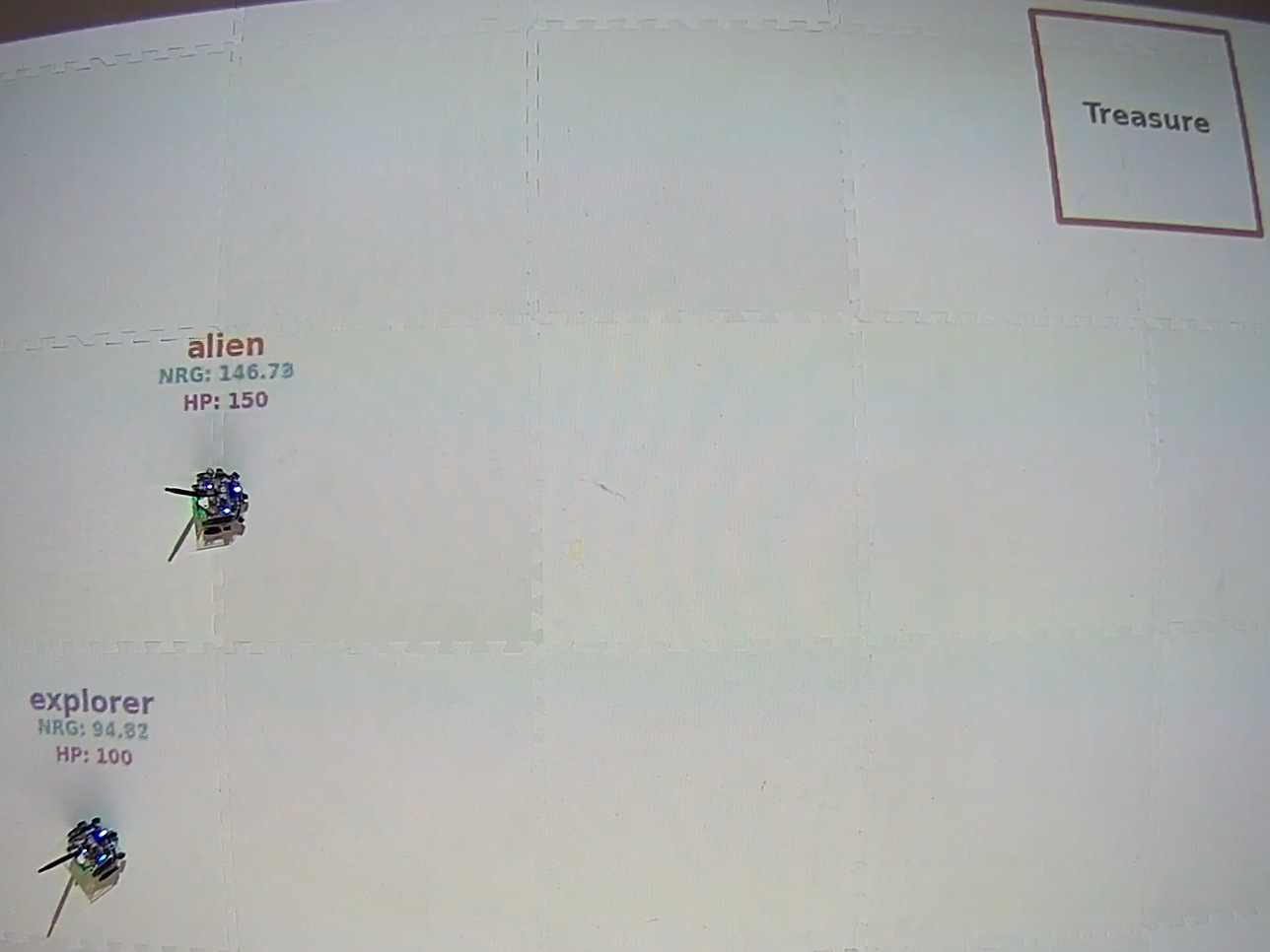}\vspace{0pt}
\caption{\small{1 explorer vs 1 alien}}
\label{fig: 1evs1a}
\end{minipage}
\begin{minipage}[b]{0.314\linewidth}
\includegraphics[width=1\textwidth]{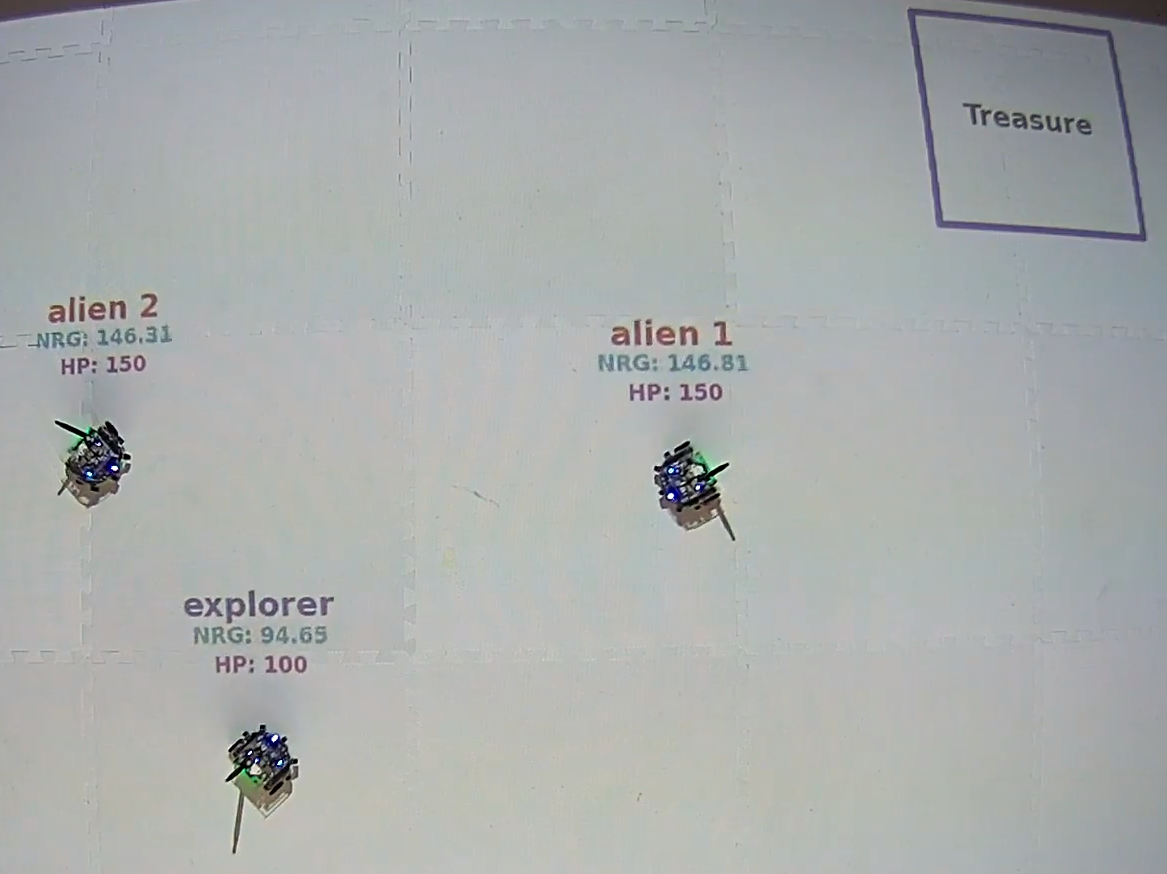}\vspace{0pt}
\caption{\small{1 explorer vs 2 aliens}}
\label{fig: 1evs2a}
\end{minipage}
\begin{minipage}[b]{0.357\linewidth}
\includegraphics[width=1\textwidth]{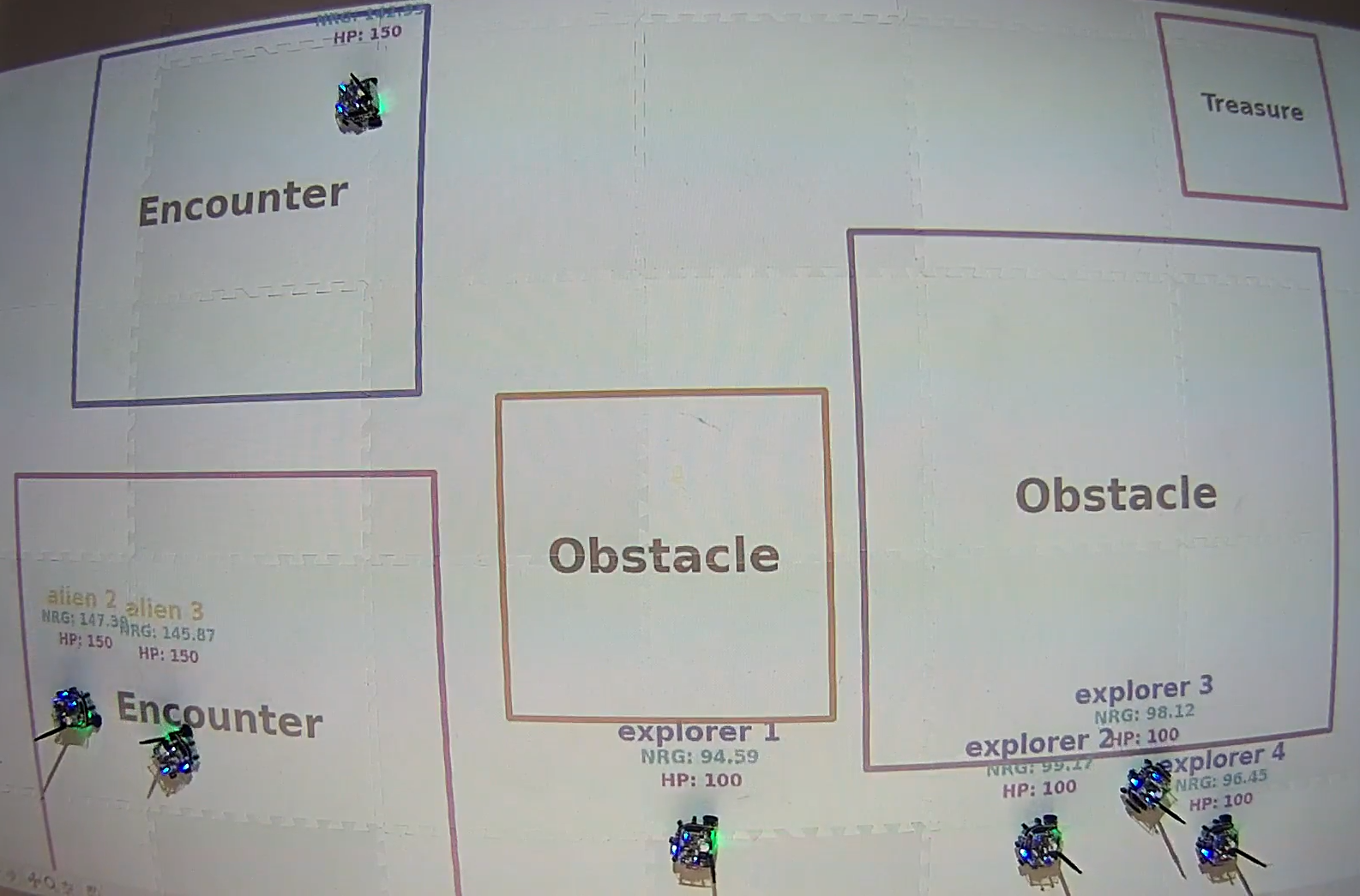}\vspace{0pt}
\caption{\small{4 explorers vs 3 aliens}}
\label{fig: 4evs3a}
\end{minipage}
% \begin{minipage}[b]{0.245\linewidth}
% \includegraphics[width=1\textwidth]{./figures/RTvideo_pictures/3pVS1e.png}\vspace{0pt}
% \caption{\small{3 pursuers vs 1 evader}}
% \label{fig: gut_fc}
% \end{minipage}
 \vspace{-5mm}
\end{figure*}
\begin{figure*}[t]
% \vspace{-4mm}
\centering
    \subfigure[Explorer Average Energy Cost]{
    \begin{minipage}[t]{0.32\linewidth}
 \includegraphics[width=1\textwidth]{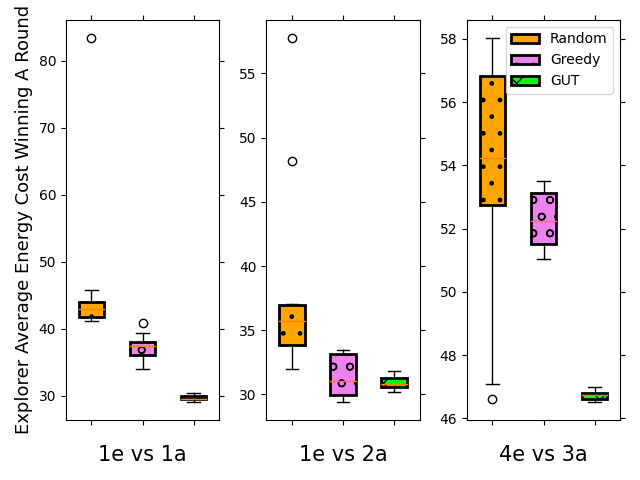}
 \vspace{-2mm}
 \label{fig: explore_energy}
 \end{minipage}}
 \subfigure[Explorer Average HP Cost]{
    \begin{minipage}[t]{0.32\linewidth}
 \includegraphics[width=1\textwidth]{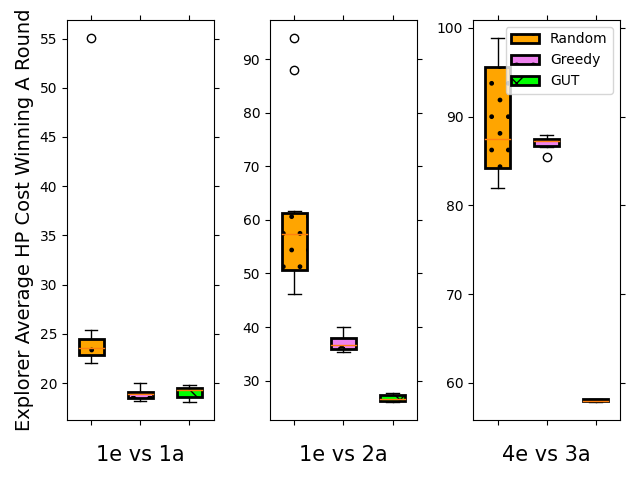}
 \vspace{-2mm}
 \label{fig: explore_hp}
 \end{minipage}}
 \subfigure[Explorer Time Cost]{
    \begin{minipage}[t]{0.32\linewidth}
    \includegraphics[width=1\textwidth]{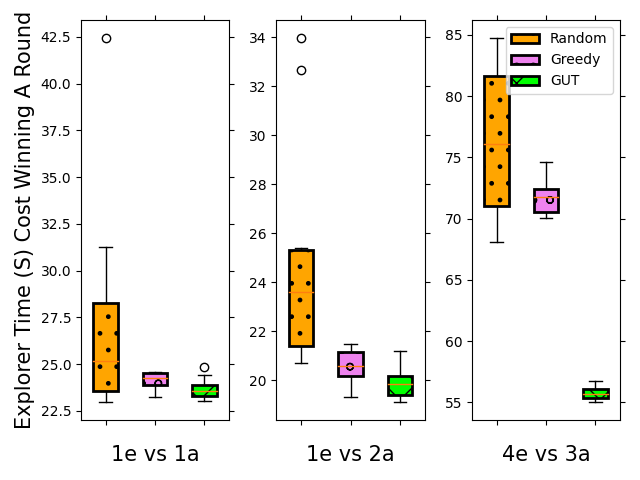}
    \vspace{-2mm}
    \label{fig: explore_time}
    \end{minipage}}
    \vspace{-2mm}
\caption{\small{The Performance Results in the Robotarium Experiments with Different Proportion of the Explore Domain.}}
\label{fig: robotarium_experiment}
\vspace{-4mm}
\end{figure*}

\section{Robotarium Experiments}
\label{sec:robotarium}
To demonstrate the GUT on the multi-robot applications, we implement our method in the Robotarium \cite{pickem2017robotarium} platform, a remote-accessible multi-robot experiment testbed that supports controlling up to 20 robots simultaneously on a 3.2m $\times$ 2.0m large rectangular area. Each robot has the dimensions 0.11 m $\times$ 0.1 m $\times$ 0.07 m in the testbed. The platform also provides a simulator helping users test their code, which can rapidly prototype their distributed control algorithms and receive feedback about their implementation feasibility before sending them to be executed by the robots on the Robotarium. 

In the Robotarium experiments, we consider three different propositions between the number of explorers and aliens in the \textit{Explore Game} domain. They are one \textit{explorer vs. one alien} (Fig. \ref{fig: 1evs1a}), \textit{one explorer vs. two aliens} (Fig. \ref{fig: 1evs2a}), and \textit{four explorers vs. three aliens} (Fig. \ref{fig: 4evs3a}). To highlight the difference between each experiment, we do not consider any obstacles in the first two scenarios. The third scenario involves (simulated) obstacle regions and two different adversarial regions (Encounters).

\begin{table}[ht]
%    \setlength{\abovecaptionskip}{3pt}
%    \scalebox{1}{
\begin{minipage}{0.48\linewidth}
 \caption{{Level 2 payoff matrix for single explorer.}}
    \label{tab:robotarium_second_level_single}
    \resizebox{\linewidth}{!}{
    \begin{tabular}{|c|c|c|}
    \hline
    \diagbox{AT}{Utility}{ET} & $\Delta$Speed & $\Delta$Direction \\ \hline
    Follow & $HP_{SF}$ & $HP_{DF}$ \\ \hline
    Retreat & $HP_{SR}$ & $HP_{DR}$ \\ \hline
 \end{tabular}}
 \end{minipage}%
 %\quad
 \hfill
 \begin{minipage}{0.48\linewidth}
 \caption{{Level 2 payoff matrix for multiple explorers.}}
\label{tab:robotarium_second_level_multiple}
    \resizebox{\linewidth}{!}{
    \begin{tabular}{|c|c|c|}
    \hline
    \diagbox{AT}{Utility}{ET} & Triangle & Diamond \\ \hline
    Follow & $HP_{TF}$ & $HP_{DF}$ \\ \hline
    Retreat & $HP_{TR}$ & $HP_{DR}$ \\ \hline
 \end{tabular}}
 \end{minipage}%
 \vspace{-2mm}
\end{table}

%In the \textit{Random} and \textit{Greedy} cases, 
Our Robotarium experiments consider four different strategies (Set $S_e$) for the explorer team:
%through the random approach and the greedy algorithm selection in various situations. 
\textit{attacking and changing direction}, \textit{attacking and changing speed}, \textit{defending and changing direction}, and \textit{defending and changing speed}. We decompose this strategy set into two levels for GUT implementation in Robotarium: Level 1 considers deciding attack or defend (Table~\ref{tab:first_level_matrix}); and Level 2 considers changing direction or speed (Table~\ref{tab:robotarium_second_level_single}) for a single explorer game while it considers triangle or diamond formation shape (Table~\ref{tab:robotarium_second_level_multiple}) for a multiple explorer game. Two different tactics payoff matrices are designed in Level 2 to differentiate the strategies between single-agent and Multi-Agent cooperation.

We compare our GUT approach with two different baseline approaches: 1) Random value selection approach (Eq. \eqref{eqn:random_value}), which chooses a strategy from a set $S_e$ by maximizing a reward function with two random variables; 2) Greedy algorithm \cite{vince2002framework} (Eq. \eqref{eqn:greedy_reward}), which maximizes the combined utilities of win rate and HP metric together (equivalent to a one-level GUT).
\footnote{Here, $s_{i_1}$ and $s_{i_2}$ represent the energy-dependent and health-dependent random reward values of strategy $i$, respectively. They follow the Gaussian distributions with different expectations. $d$ is the distance between explorer and goal point. $n_a$ and $u_{hp}$ are the number of active aliens and their average unit attacking damage cost. $c_1$ and $c_2$ are the corresponding coefficients. $W_i$ and $HP_i$ are the winning and HP utility values of strategy $i$ from Sec.~\ref{sec:domain}.}
\begin{eqnarray}
\begin{split}
    & s_e^* = \mathop{\argmax_{i \in S_e}}\mathop{\left[100 - c_1 \cdot s_{i_1} \cdot d - c_2 \cdot s_{i_2} \cdot n_a \cdot u_{hp}\right]}     \label{eqn:random_value}\\
    % & E_{uc} = HP_{uc} \times \beta_{uc_0} + \varepsilon; \\
    % & E_{el} = 100 - HP_{asc} \times \beta_{asc_0} + \varepsilon.
    & s_e^* = \mathop{\argmax_{i \in S_e}}\mathop{\left[W_{i} \cdot HP_i \right]}
    % & E_{uc} = HP_{uc} \times \beta_{uc_0} + \varepsilon; \\
    % & E_{el} = 100 - HP_{asc} \times \beta_{asc_0} + \varepsilon.
    \label{eqn:greedy_reward}
\end{split}
\end{eqnarray}

We implement each case with real robots in the Robotarium and conduct ten simulation trials (rounds) for each scenario in the Robotarium simulator. 
We initialize each robot in the same group (explorer or adversary) with equal energy (battery) levels and health points (HP); for example, explorer with 100 points and alien with 150 points in Energy and HP correspondingly. Also, we assumed that every moving step and attack damage cost 0.1\% of energy and 0.3\% of HP, respectively. Other settings are similar to the previous experiments.

\paragraph*{Results and Discussion}
Fig.~\ref{fig: robotarium_experiment} presents the results of the Robotarium experiments. The average costs of energy, HP, and the time taken to complete the mission of explorer are shown in Figs.~\ref{fig: explore_energy}, \ref{fig: explore_hp}, and \ref{fig: explore_time}, respectively. The data shows that in the single explorer cases (scenarios \textit{1e vs. 1a} and \textit{1e vs. 2a}), both GUT and greedy stand out compared to the random approach, but the GUT is not significantly better than the greedy approach consistently in all the performance metrics. We attribute this to the fact that GUT uses only two levels of action decomposition and therefore does not offer significant advantages compared to the one-level GUT/greedy approach in this simple scenario.
%To some extent, the average execution times of the \textit{random} and \textit{greedy} present better performance in Fig. \ref{fig: explore_time}. 
However, for the multiple explorer case (scenario \textit{4e vs. 3a}), the GUT shows superior performance over other methods in all metrics, demonstrating that the GUT can help MAS organize their behaviors and select a suitable strategy adapting to complex situations. We can get a similar observation from the perspective of the winning rate and the average number of explorers lost from the Table. \ref{tab: robot_wrc}.

%\begin{center}
%\scalebox{0.79}
%{
\begin{table}
\caption{Performance results in Robotarium experiments.}
\label{tab: robot_wrc}\resizebox{\columnwidth}{!}{\begin{tabular}{|c|c|c|c|c|c|c|}
    \hline
    \multirow{2}*{\diagbox{PRD}{APP}} & \multicolumn{3}{c|}{Winning Rate} &  \multicolumn{3}{c|}{Lost Explorers Per Round/win} \\
    \cline{2-7}
    ~&Random&Greedy&GUT&Random&Greedy&GUT \\
    % \multirow{2}*{Greedy} & \multirow{2}*{GUT} \\
    \hline
    1e vs 1a & 100\% & 100\% & 100\% & - & - & - \\ 
    \hline
    1e vs 2a & 90\% & 100\% & 100\% & 0.1 & - & - \\ 
    \hline
    4e vs 3a & 50\% & 90\% & 100\% & 2.8 & 2.0 & - \\ 
    \hline
\end{tabular}}
\vspace{-4mm}
\end{table}

Generally speaking, by demonstrating the \textit{Explore} domain in the Robotarium, we further proofed that the GUT could help the intelligent agent (robot) rationally analyze different situations in real-time, effectively decompose the high-level strategy into low-level tactics, and reasonably organize groups' behaviors to adapt to the current scenario. From the system perspective, through applying the GUT, the group can present more complex strategies or behaviors to solve the dynamic changing issues and optimize or sub-optimize the group utilities in MAS cooperation. From the individual agent perspective, GUT reduces the agent's costs and guarantees sustainable development for each group member, much like human society.

\section{Conclusion}
\label{sec:conclusions}
%  \vspace{-2mm}

We introduce a new Game-theoretic Utility Tree (GUT) for Multi-Agent decision-making in adversarial scenarios and present an example real-time strategy game called \textit{Explore Domain}.
Through extensive numerical simulations, we analyze GUT and compare it against a state-of-the-art cooperative decision-making approach, such as the greedy selection method in QMIX. We verified the effectiveness of GUT through two types of experiments involving interaction and information prediction between the agents. The results showed that the GUT could organize more complex relationships among MAS cooperation, helping the group achieve more challenging tasks with lower costs and higher winning rates. 
Further, we proofed the effectiveness of the GUT in the real robot application through the implementation of the \textit{Explore Domain} on the Robotarium.
It demonstrated that the GUT could effectively organize Multi-Agent cooperation strategies, enabling a group with fewer advantages to achieve higher performance.

\bibliographystyle{ACM-Reference-Format}
\bibliography{sample-bibliography} 

\end{document}